\def\paperversion{arxiv}

\newif\ifshowappendix
\newif\ifanonymous
\newif\iftitleparen
\newif\iftitleprefix
\newif\ifvenueinfo
\newif\ifrelaxedformat

\def\pvconf{conf}\def\pvsupp{supp}\def\pvarxiv{arxiv}
\ifx\paperversion\pvconf
  \showappendixfalse \anonymoustrue  \titleparenfalse \titleprefixfalse \venueinfotrue  \relaxedformatfalse
\else\ifx\paperversion\pvsupp
  \showappendixtrue  \anonymoustrue  \titleparentrue  \titleprefixfalse \venueinfotrue  \relaxedformatfalse
\else\ifx\paperversion\pvarxiv
  \showappendixtrue  \anonymousfalse \titleparenfalse \titleprefixtrue  \venueinfofalse \relaxedformattrue
\else
  \errmessage{Unknown \string\paperversion=`\paperversion' -- use conf, supp, or arxiv}%
\fi\fi\fi

\def\titleprefixtext{\texorpdfstring{\ourname}{TD-Orch}: }

\ifanonymous
  \documentclass[sigplan,10pt,anonymous,screen]{acmart}
\else
  \documentclass[sigplan,10pt,screen]{acmart}
\fi
\ifrelaxedformat
  \makeatletter
  \def\@titlefont{\huge\bfseries}
  \makeatother
\fi

\usepackage{amsmath,amssymb,amsfonts}
\usepackage{algorithmic}
\usepackage{graphicx}
\usepackage{textcomp}
\usepackage{xcolor}
\usepackage{url}

\usepackage{tikz}

\usepackage{filecontents}
\usepackage{xspace}
\usepackage{balance}
\usepackage{enumitem}
\usepackage[ruled,lined,linesnumbered,noend]{algorithm2e}
\usepackage{cleveref}
\usepackage{multirow}
\usepackage{afterpage}
\usepackage{siunitx}
\usepackage{tabularx}
\usepackage{booktabs}
\newcolumntype{Y}{>{\raggedleft\arraybackslash}X}
\usepackage{listings}
\usepackage{subcaption}
\usepackage{threeparttable} 
\usepackage{caption} 
\usepackage{placeins}
\usepackage{afterpage}
\usepackage{amsthm} 
\newtheorem{theorem}{Theorem} 
\newtheorem{lemma}{Lemma}

\usepackage{float}
\usepackage{bm}
\usepackage{listings}

\startPage{1}
\setcopyright{none}
\renewcommand\footnotetextcopyrightpermission[1]{}

\newcommand{\ourname}{TD-Orch\xspace}
\newcommand{\ourgraph}{TDO-GP\xspace}

\newcommand{\kv}{key-value\xspace}

\newcommand{\KV}{Key-Value\xspace}
\newcommand{\taskdata}{task-data\xspace}
\newcommand{\Taskdata}{Task-data\xspace}
\newcommand{\TaskData}{Task-Data\xspace}
\newcommand{\naive}{na\"ive\xspace}

\newcommand{\Topt}{\mbox{$T_{opt}$}}

\ifdefined\revisionmode
\newcommand{\hongbo}[1]{{\color{brown} {\bf Hongbo:} #1}}
\newcommand{\charlie}[1]{{\color{red} {\bf Charlie:} #1}}
\newcommand{\laxman}[1]{{\color{magenta} {\bf Laxman:} #1}}
\newcommand{\guy}[1]{{\color{purple} {\bf Guy:} #1}}
\newcommand{\guyup}[1]{{\color{purple} #1}}
\newcommand{\phil}[1]{{\color{blue} {\bf Phil:} #1}}
\newcommand{\yan}[1]{{\color{olive} {\bf Yan:} #1}}
\newcommand{\yiwei}[1]{{\color{green} {\bf Yiwei:} #1}}
\newcommand{\qiushi}[1]{{\color{orange} {\bf Qiushi:} #1}}
\else
\newcommand{\hongbo}[1]{}
\newcommand{\charlie}[1]{}
\newcommand{\laxman}[1]{}
\newcommand{\guy}[1]{}
\newcommand{\guyup}[1]{}
\newcommand{\phil}[1]{}
\newcommand{\yan}[1]{}
\newcommand{\yiwei}[1]{}
\newcommand{\qiushi}[1]{}
\fi

\definecolor{forestgreen}{rgb}{0.13, 0.55, 0.13}

\renewcommand{\emph}{\textit}

\ifshowappendix

\else
\ifdefined\nocolor

\else

\fi
\fi

\newcommand{\confversiononly}[1]{}

\ifx\paperversion\pvarxiv
  \newcommand{\nonarxivonly}[1]{}
\else
  \newcommand{\nonarxivonly}[1]{#1}
\fi

\ifx\paperversion\pvarxiv
  \makeatletter
  \def\country#1{\global\@ACM@countrypresenttrue}
  \makeatother
\fi

\usepackage{refcount}
\makeatletter
\@namedef{sr@pre@sec}{\S}
\@namedef{sr@pre@subsec}{\S}
\@namedef{sr@pre@subsubsec}{\S}
\@namedef{sr@pre@appendix}{\S}
\@namedef{sr@pre@appendix_subsubsec}{\S}
\@namedef{sr@pre@tab}{Table\,}
\@namedef{sr@pre@Tab}{Table\,}
\@namedef{sr@pre@fig}{Fig.\,}
\@namedef{sr@pre@Fig}{Fig.\,}
\@namedef{sr@pre@alg}{Alg.\,}
\@namedef{sr@pre@Alg}{Alg.\,}
\newcommand{\SetSuppRef}[2]{\@namedef{sr@val@#1}{#2}}
\ifshowappendix
  \newwrite\sr@out
  \immediate\openout\sr@out=suppref-values.tex
  \immediate\write\sr@out{\@percentchar\space Auto-generated by the supp/arxiv build of main.tex; do not edit by hand.}%
  \def\sr@dispatch#1:#2\@nil#3{%
    \@ifundefined{sr@pre@#1}{\Cref{#3}}{\csname sr@pre@#1\endcsname\ref{#3}}}
  \newcommand{\suppref}[1]{%
    \immediate\write\sr@out{\noexpand\SetSuppRef{#1}{\getrefnumber{#1}}}%
    \sr@dispatch#1:\@nil{#1}}
\else
  \InputIfFileExists{suppref-values.tex}{}{}%
  \def\sr@dispatch#1:#2\@nil#3{%
    \@ifundefined{sr@val@#3}{the supplementary material}%
      {\@ifundefined{sr@pre@#1}{}{\csname sr@pre@#1\endcsname}%
       \csname sr@val@#3\endcsname\ (supplementary material)}}
  \newcommand{\suppref}[1]{\sr@dispatch#1:\@nil{#1}}
\fi
\makeatother

\newcommand{\whp}{\textit{whp}\xspace}

\newcommand{\RNum}[1]{\expandafter{\romannumeral #1\relax}}
\newcommand\romenum[1]{\mbox{(\textit{\RNum{#1}})}\nolinebreak{}}

\newcommand{\myparagraph}[1]{\smallskip\noindent {\bf #1.}}

\newcommand{\hide}[1]{}

\newtheorem{thm}{Theorem}
\newtheorem{dfn}[thm]{Definition}

\newcommand{\VERTEXSUBSET}{\textsc{VertexSubset}\xspace}
\newcommand{\DISTVERTEXSUBSET}{\textsc{DistVertexSubset}\xspace}
\newcommand{\EDGEMAP}{\textsc{EdgeMap}\xspace}
\newcommand{\DISTEDGEMAP}{\textsc{DistEdgeMap}\xspace}

\SetKwProg{Struct}{Struct}{}{end}

\begin{document}

\title{\iftitleprefix\titleprefixtext\fi
Efficient Task-Data Orchestration for Distributed Systems with Application to Graph Processing%
}

\author{Yiwei Zhao}
\authornote{Equal contribution; order of appearance determined randomly.}
\affiliation{%
  \institution{Carnegie Mellon University}
  \country{USA}}
\author{Qiushi Lin}
\authornotemark[1]
\affiliation{%
  \institution{University of Texas at Austin}
  \country{USA}}
\author{Hongbo Kang}
\affiliation{%
  \institution{Tsinghua University}
  \country{China}}
\author{Guy E. Blelloch}
\affiliation{%
  \institution{Carnegie Mellon University}
  \country{USA}}
\author{Laxman Dhulipala}
\affiliation{%
  \institution{University of Maryland}
  \country{USA}}
\author{Yan Gu}
\affiliation{%
  \institution{University of California, Riverside}
  \country{USA}}
\author{Charles McGuffey}
\affiliation{%
  \institution{Reed College}
  \country{USA}}
\author{Phillip B. Gibbons}
\authornote{Corresponding author: \href{mailto:gibbons@cs.cmu.edu}{gibbons@cs.cmu.edu}}
\affiliation{%
  \institution{Carnegie Mellon University}
  \country{USA}}

\begin{abstract}

We introduce a task-data orchestration abstraction that supports a range of distributed applications.
Given a batch of \textit{lambda} tasks each requesting a data item, where both tasks and data are distributed across multiple machines, each task must be co-located with its target data (by moving tasks and/or data) and then executed.
We present TD-Orch, an efficient lambda-task-centric orchestration framework for low-overhead load balancing with a simple interface for application developers.
TD-Orch employs a distributed push-pull technique, leveraging the bidirectional flow of both tasks and data to achieve load balance across machines even under highly skewed access patterns, with low communication overhead.
Experimental results on key-value stores show that TD-Orch achieves up to $2.8\times$ speedup over existing distributed scheduling baselines.
Building on TD-Orch, we present TDO-GP, a distributed graph processing system with $4.1\times$ average speedup over state-of-the-art open-source distributed graph systems for general graph processing.

\end{abstract}

\maketitle
\pagestyle{plain}

\section{Introduction}
\label{sec:introduction}
Distributed computer systems, including multi-socket NUMAs~\cite{lameter2013numa}, supercomputers~\cite{schneider2022exascale}, disaggregated memory~\cite{intelagilex,calciu2021rethinking,zeng2025performance,alias2010optimizing,guo2026cxlmc}, processing-in-memory~\cite{upmem,mutlu2023primer}, computational storage~\cite{samsungsmartssd,gu2016biscuit} and datacenters~\cite{wu2015cloud,kumar2021resource} provide tremendous compute power and memory capacity.
Motivated by these platforms, there has been a huge body of work on building high-performance distributed systems for applications such as graph processing~\cite{morales2021graph,zhang2017simple,welc2013greenmarl,jin2023graphsycl,peng2018graphphi,busato2018hornet,veras2016scale}, machine learning~\cite{moritz2018ray,narayanan2019pipedream,thorpe2021dorylus,abadi2016tensorflow,chen2015mxnet,dagli2022axonn,baskaran2019fast}, \kv stores~\cite{decandia2007dynamo,escriva2012hyperdex,chang2006bigtable,cooper2008pnuts,ye2021hardware,basin2017kiwi,daly2018numask}, databases~\cite{corbett2013spanner,shute2013f1,dageville2016snowflake,kim2025no}, among many other application areas.

Writing correct and fast distributed applications is, however, notoriously challenging.
Though writing \textit{low-level code} with interfaces such as Message Passing Interface (MPI)~\cite{gropp1999using} can deliver high performance, it places the burden of performance engineering on programmers, substantially complicating programming and system development---as shown by manual efforts in scheduling~\cite{isard2009quincy,ousterhout2013sparrow,popovici2018fft}, MPI communication~\cite{vadhiyar2000automatically,almasi2005optimization,faraj2005automatic,faraj2006star,peng2023sharedmem,feng2022dragonfly,noeth2009scalatrace,chatterjee2013integrating,vo2010scalable,lange2018static}, data placement~\cite{wang2023provably,curino2010schism,kumar2014sword,ren2021sentinel}, caching~\cite{jin2017netcache,jiang2002lirs,decandia2007dynamo,chandra2007paxos,suh2004dynamic,huang2016c3d}, etc.

Many system designs written in low-level interfaces repeatedly reinvent techniques to address common challenges.
Among them, a commonly-seen challenge is achieving low-overhead \textit{load balance} in both computation and inter-machine communication.
Many \textit{general-purpose} high-level frameworks \cite{dean2008mapreduce,zaharia2016apache,moritz2018ray,rocklin2015dask,isard2007dryad,murray2011ciel,murray2013naiad,carbone2015flink} do not address this directly: users must still manage data placement and communication to avoid hot spots.
Consequently, prior \textit{domain-specific} systems often adopt ad hoc solutions, e.g., mirror nodes in graphs~\cite{zhu2016gemini,gonzalez2012powergraph}, virtual servers in \kv stores~\cite{stoica2003chord,decandia2007dynamo}, or heuristic partitioning in databases~\cite{curino2010schism,quamar2013sword}.
Despite varied forms, \textit{these techniques share similar high-level ideas}, suggesting the possibility of one unified abstraction/framework.

We achieve low-overhead load balancing with an efficient scheduling scheme and a unified interface for one commonly-seen family of applications.
We build upon prior work on *-centric (e.g., node-centric, edge-centric) abstractions by introducing a \textit{\textbf{\TaskData Orchestration}} abstraction based on microservices-inspired \textit{lambda-tasks}.
Our lambda-task-centric abstraction (formally defined in \S\ref{subsec:abstraction_interface}) can efficiently support a wide range of distributed applications, including graph processing, \kv stores, databases, and others (as discussed in \S\ref{subsec:abstraction_application}).
Users or application developers only need to specify basic operation(s) each lambda-task performs. 
An underlying system automates efficient orchestration of computation and communication, allowing users to write high-performance code with less system-management effort.

\myparagraph{System Design}
In this paper, we propose \textbf{\textit{\ourname}}, a system that efficiently realizes our task-data orchestration abstraction and targets low-overhead load balance (\S\ref{sec:main_design}).
Its key technique is to orchestrate the system through \textbf{\textit{distributed push-pull}}, which exploits both \textit{pushing} tasks to data and \textit{pulling} data to tasks, unlike most prior approaches using only one direction~\cite{vora2014aspire,li2020pegasus,gog2016firmament,fan2011small,murray2013naiad,jelasity2005gossip,vadhiyar2000automatically,menon2013distributed,leitao2007epidemic}.
As push or pull alone cannot avoid hotspots, \ourname detects hot data and introduces \romenum{1} \textit{transit} machines inside a \textit{communication forest} routing scheme and \romenum{2} mergeable \textit{meta-task} structures.
Together, these ensure that low-contention data are handled by directly pushing task contexts to the data, while high-contention data are delivered to tasks communi\-cation-efficiently, achieving load balance with low overhead.

Although our design may resemble prior HPC collectives such as reduction trees~\cite{vadhiyar2000automatically} and gossip-style aggregation~\cite{jelasity2005gossip}, it is fundamentally different.
\romenum{1} Unlike prior empirical tree-based designs, we introduce a \textbf{\textit{bidirectional push-pull flow}} over the forest with provable optimality guarantees that prior unidirectional tree-structured designs lack.
\romenum{2} Compared with prior theory-guided approaches (e.g.,~\cite{axtmann2015practical}), our design requires only two passes over the tree/forest rather than at least three.
\romenum{3} Moreover, we cast scheduling as a unified \taskdata orchestration abstraction supporting a variety of applications and distributed systems platforms.
We illustrate the benefits of \ourname on two case study applications.

\myparagraph{Case Study I: Distributed \KV Stores}
\ourname achieves strong performance in distributed \kv stores.
On a 1,024-core cluster, our design achieves geomean speed\-ups of $2.09\times$, $1.42\times$, and $2.83\times$ over three baselines (\S\ref{sec:key_value_case_study}).

\myparagraph{Case Study II: Distributed Graph Processing}
In our primary case study, we develop \textbf{\ourgraph}, a distributed graph processing system atop \ourname.
\ourgraph leverages \ourname as the primary mechanism for orchestrating \taskdata flow and employs additional implementation techniques to fully exploit the \taskdata layouts provided by \ourname.
Like \ourname, \ourgraph is theoretically grounded, offering work-efficient bounds for the five graph algorithms implemented in this study.
In addition to directly leveraging \ourname, we further optimize global communication, local computation, and the coordination between system components, tailoring these improvements for graph workloads to fully exploit the execution layout provided by \ourname.
\ourgraph outperforms two major families of prior graph systems: compared with linear-algebra-based systems (e.g.,~\cite{mofrad2020graphite,ahmad2018la3,anderson2016graphpad}), it provides significantly better provable and empirical \textit{load balance} by using transit machines; and compared with graph-algorithm-based systems (e.g.,~\cite{zhu2016gemini,chen2019powerlyra,yan2015effective}), it achieves better graph-theoretic guarantees and higher performance using a superior \textit{graph-tailored implementation}.
\ourgraph achieves geomean speedups of $4.1\times$ over the best prior open-source distributed graph systems targeting general graph problems~\cite{zhu2016gemini,mofrad2020graphite,ahmad2018la3,yan2015effective,behera2024starplat} and up to $15\times$ on specific real-world datasets (\S\ref{sec:graph_processing_system}--\ref{sec:eval}).
Moreover, the interface enables user development of complex graph algorithms in 70 lines of C++ code, versus 400+ lines in prior work (e.g.,~\cite{zhu2016gemini}).

\myparagraph{Main Contributions}
In summary, in this paper:
\begin{itemize}[topsep=0pt,itemsep=0pt,parsep=0pt,leftmargin=15pt]
    \item We define \textit{\taskdata orchestration} as a fundamental execution pattern in distributed applications, with a concrete lambda-centric interface (Fig.\ref{fig:orchestration_interface}) that reduces development effort for users to write high-performance code.
    \item We design \textit{\ourname}, a theoretically grounded and practically efficient orchestration system that supports the interface and achieves load balance with low overhead, and show its effectiveness for \kv stores.
    \item We develop \textit{\ourgraph}, a distributed graph processing system built on \ourname and enhanced with tailored implementation techniques, which significantly outperforms prior state-of-the-art graph systems~\cite{zhu2016gemini,mofrad2020graphite,ahmad2018la3,yan2015effective,behera2024starplat}.
\end{itemize}

\section{Preliminaries}
\label{sec:setup}

\subsection{Problem Definition and System Modeling}
\label{subsec:abstraction_system_model}

\myparagraph{Tasks and Data}
We consider a batch of $n$ tasks, each with a context size of at most $\sigma$ words, to be executed during an application.
At initialization, the tasks are evenly distributed across the machines.
Data are partitioned into \textit{data chunks} of granularity $B$ words, analogous to cache-line/page storage.
Data are moved between machines in units of data chunks.
Each chunk starts on a random machine, providing adversary-resistant load balance in both storage and access patterns~\cite{sanders1996competitive}.
Such randomized placement is commonly used for robustly high performance~\cite{wang2023provably,kang2022pimtreevldb,shan2023explore,luo2023smart}.
During each orchestration stage, each task requires at most one data chunk (possibly zero) stored outside its task context.
A task completes only after this data chunk is co-located on the same machine, processed and, if necessary, written back.

\myparagraph{Design Goals}
Both tasks and data chunks are allowed to be moved between machines, and even routed through intermediate machines.
Our goal is to \textit{orchestrate} tasks and data---i.e., to design an online scheduling scheme that \romenum{1} minimizes the amounts of communication and computation, and \romenum{2} ensures load balance across machines to maximize throughput, even for skewed workloads with hot data chunks.

\myparagraph{System Modeling}
For our analysis results, we use the popular Bulk-Synchronous Parallel (BSP) model~\cite{valiant1990bridging}.
The model assumes a system with $P$ machines, each with a processing unit and a local memory without shared global memory.
All machines are connected by a point-to-point network, similar to the MPI abstraction.
The system operates in \textit{rounds} (or \textit{supersteps}),
where periodic barrier synchronizations are required.
BSP accounts for \textbf{\textit{computation time}} (maximum work) and \textbf{\textit{communication time}} (maximum bytes sent/received) across machines in each round (see \suppref{sec:model} for details).
The model captures the importance of load balance in its use of ``maximum'' in these metrics.

\subsection{Prior Scheduling Strategies}
\label{subsec:preliminary_prior_work}

\Taskdata orchestration seeks a scheduling scheme in which $n$ task contexts, each of size $O(\sigma)$, access up to $n' \leq n$ distinct data chunks, each of size $B$.
The ideal communication time of $\Topt=O(n'\cdot\min\{B, \sigma\} / P)$ is unachievable in practice using only online scheduling algorithms.
Instead, prior work typically adopts one of the following strategies.

\myparagraph{Direct Pull}
A \naive strategy first eliminates duplicate data chunk requests by any single machine, then fetches all required data to corresponding tasks.
This approach, common in remote direct memory accesses (RDMA), works well when \textit{contention} (reference count)
on each chunk is low.
However, contention can be high---for example, in graph processing, contention can arise from extremely high-degree vertices---leading to high loads on machines storing ``hot'' chunks and severe load imbalance.
In the worst case, all $P$ machines request the same data chunk ($n'=1$), yielding communication time $O(PB)$ and competitive ratio $O(PB/\Topt) =O\left(P^2B/\min\{B,\sigma\}\right)$, which is \textbf{\textit{prohibitive}} as $P$ grows.

\myparagraph{Direct Push}
Another common strategy in distributed settings, e.g., in remote procedure calls, is to offload tasks to the machines that store the required data chunks.
However, machines holding ``hot'' data chunks incur high communication cost (for receiving tasks) and high computation work (for executing tasks), resulting in severe load imbalance.
In the worst case, all $n$ tasks are pushed to the same chunk,
resulting in $O(n\sigma)$ communication time and a competitive ratio of $O\left(P n\sigma/\min\{B, \sigma\}\right)$, which is also \textbf{\textit{prohibitive}}.

\myparagraph{Theory-Guided Designs}
Prior work on the Massively Parallel Computing (MPC) model~\cite{im2023massively} can extend to theoretically-efficient solutions to the orchestration problem (e.g.,~\cite{im2023massively,goodrich2011sorting,chowdhury2013oblivious}).
The procedure begins by sorting tasks according to the addresses of their required data chunks.
Then, MPC broadcasting~\cite{ghaffari2019massively} is used to distribute each data chunk to the corresponding tasks.
After task execution, a reverse broadcasting step writes the results back to the data chunks, followed by a reverse sorting step that restores tasks to their original order.
It has been shown that the computation work is load balanced, and the communication time is $\Theta(\frac{n}{p}\log_{\frac{n}{P}}{P})$~\cite{fish2015computational}, which is optimal~\cite{roughgarden2018suffles}.
Although MPC schemes are asymptotically optimal, we show in \S\ref{sec:key_value_case_study} that their constant factors are non-negligible in practice, resulting in \textbf{\textit{poor performance}} on real-world systems (e.g.,~\cite{axtmann2015practical}) compared to empirical designs.

\myparagraph{Empirical Designs}
Many empirical designs (e.g.,~\cite{vora2014aspire,li2020pegasus,gog2016firmament,fan2011small,murray2013naiad,jelasity2005gossip,vadhiyar2000automatically,menon2013distributed,willcock2010ampp,ghosh2022neighborhood,chen2016baymax,kim2024tccl,scogland2012heterogeneous,pande1995scalable,kong2016pipes,nemirovsky2017machine,xu2010mitigating,popov2019efficient,moritz2018ray,rocklin2015dask,yan2019thinker}, \suppref{sec:appendix_related_frameworks}), while not directly targeting \taskdata orchestration, introduce useful techniques, such as communication aggregation.
But few provide rigorous analysis of competitive ratios, and none offer tight bounds, limiting their (i) generalizability to broader distributed systems and (ii) robustness to skewed workloads.

\section{Lambda-Task-Centric Abstraction}
\label{sec:interface}

\subsection{Abstraction and Interface}
\label{subsec:abstraction_interface}

Our \taskdata orchestration abstracts applications that follow an execution pattern of repeated BSP-style~\cite{valiant1990bridging, valiant1990general,valiant1988optimally} \textit{orchestration stages}.
In each stage, a batch of \textit{lambda-tasks} (simply \textit{tasks} hereafter) runs in parallel.
Each task resembles a C++ closure (Fig.\ref{fig:orchestration_interface}): it contains local metadata (\textit{task context}, or \textsc{LocalContexts}), an \textsc{InputPointer} specifying the external data to read, a local lambda function $f$ describing the task's computation, and an \textsc{OutputPointer} specifying where the result will be written.
Within a stage, each task \romenum{1} \textit{reads} the input data chunk from \textsc{InputPointer}, \romenum{2} performs a local \textit{computation} by applying $f$ to its local contexts and the chunk, and \romenum{3} \textit{writes} the result to \textsc{OutputPointer}, using the \textsc{wb} function (which specifies how multiple updates to the same data chunk are ``merged'').

\begin{figure}
\begin{lstlisting}
function Orchestration(
    tasks: An array of Task structures
    wb: function, // wb(pointer, value)
);
struct Task {
    pointer InputPointer, OutputPointer;
    LocalContexts[]; // Per-Task local metadata
    Lambda f; // f(InValue) -> OutValue
};
\end{lstlisting}
\vspace{-1em}
\caption{The \textit{\TaskData Orchestration} interface. Task and data are distributed among the same set of machines.}
\label{fig:orchestration_interface}
\vspace{-1em}
\end{figure}

\subsection{Supported Applications}
\label{subsec:abstraction_application}

Our abstraction focuses on the case where each lambda task requires at most one external data chunk, simplifying both the interface and the orchestration system supporting it.
While this might seem limiting, in fact many distributed applications can directly leverage the interface to reduce the manual burden of system management, while achieving high performance.
Here we show multiple examples.

\myparagraph{Graph Processing}
Many vertex- and edge-centric graph algorithms share a common execution pattern: a vertex array stores intermediate values; in each stage, a subset of edges are activated, read the values of their endpoints, perform a computation using these vertex values and the edge weight, and write the result back to the vertex array.
Examples include breadth-first search (BFS), shortest paths, betweenness centrality, connectivity, PageRank, and spanning forests (\S\ref{sec:graph_processing_system}).

All such graph applications can be expressed within our \taskdata orchestration framework.
An example pseudocode\footnote{This example shows how graph algorithms can be written using the orchestration interface. It is intentionally simple but not work-efficient. Work-efficient implementations rely on auxiliary structures introduced later in \S\ref{sec:graph_processing_system}, while retaining a similarly simple programming interface (Fig.\ref{fig:distedgemap}).}
for BFS is shown in Alg.\ref{alg:bfs}.
It proceeds in stages, maintaining a frontier of vertices $i$ steps away from the source in the $i$-th stage. 
Here, tasks correspond to edges, and data correspond to distance values of vertices.
At each orchestration stage, lambda $f$ is applied to each edge.
It checks whether the source vertex of the edge was activated in the previous round, having a distance of $\textsc{round}-1$.
If so, it writes distance \textsc{round} to the destination vertex, if it was not written already.

In Alg.~\ref{alg:bfs}, our interface lets users focus solely on application semantics, freeing them from performance optimization.

\begingroup
\newcommand{\TightAlgoSkip}{\vspace{-3pt}}
\SetAlgoSkip{TightAlgoSkip}
\begin{algorithm}[t]
\small
\fontsize{9pt}{12pt}\selectfont
\SetKwFor{IF}{If}{then}{endif}
\SetKwFor{While}{While}{do}{endwhile}
\SetKwFor{ParFor}{parallel for}{do}{endfch}
\SetKwFor{For}{for}{do}{endfch}
\SetKwFor{Function}{Function}{:}{}
\KwIn{$G=(V,E)$, $n=|V|$, $m=|E|$; \textsc{StartVertex}}
\KwOut{\textsc{Dist}[$1:n$]: Distributed vertex value array}
Initialize vertex value array: \textsc{Dist}[$1:n$] = [-1]$\times$n\\
\textsc{VertexUpdated} = 1; \textsc{round} = 0\\
\Struct{Task}{%
    LocalContext: Edge (u, v)\\
    InputPointer, OutputPointer: ptr(dist[u]), ptr(dist[v])\\
    \Function{f\textnormal{(value\_u: int)}}{
        \Return \textsc{round} \textbf{{If}} (value\_u == \textsc{round} - 1) \textbf{{Else}} -1
    }
}
\Function{\textsc{WriteBack}\textnormal{(dist\_ptr: pointer, value\_v: int)}}{
    \IF{\textnormal{*dist\_ptr == -1} and \textnormal{value\_v != -1}}{
        *dist\_ptr = value\_v;\quad\textsc{VertexUpdated} = 1
    }
}
\Function{BFS\textnormal{($G$,} \textsc{StartVertex}\textnormal{: int)}}{
    \textsc{Dist}[\textsc{StartVertex}] = $0$;\quad Initialize \textsc{tasks}[1:m] with $E$\\
    \While{\textsc{VertexUpdated} $>$ \textsc{0}}{
        \textsc{round} += $1$; \textsc{VertexUpdated} = 0 \\
        \textsc{Orchestration}(\textsc{tasks}, \textsc{WriteBack}) \ // tasks, wb  
    }
    \Return \textsc{Dist}[$1:n$]
}
\caption{Breadth-First Search (BFS)}\label{alg:bfs}
\end{algorithm}
\endgroup

\myparagraph{\KV Stores}
Many \kv stores naturally fit within the \taskdata orchestration model.
A straightforward example is a (distributed) hash table, where reading and updating a batch of items can be expressed as a one-stage orchestration.
Here, the tasks correspond to operations by defining $f$ as the per-item operation, and the data reside in the hash-table storage.
Each batch can be executed using a single orchestration stage that:
\romenum{1} reads values from the hash table using task-specified keys,  
\romenum{2} processes fetched values, and
\romenum{3} writes results back to the hash table.
\S\ref{sec:key_value_case_study} provides further details.

Similarly, many ordered indexes (e.g., B-trees, radix trees) can be represented within this abstraction by decomposing each search/update query into multiple orchestration stages.
In each stage, a single level of the index/tree is traversed: the corresponding internal node is retrieved, and $f$ performs the search within the node to identify the appropriate child.
This process continues until the leaf is reached.

\myparagraph{Transactional/Analytical Databases}
OLTP/OLAP share structural similarities with \kv stores and can process each query as a sequence of orchestration stages.
Tree-index operations, such as Masstree~\cite{mao2012cache} traversals and range scans, access one node/block per level, with each level handled in one stage.
Version chains resemble bucket-based hash tables, while record storage can use direct hash tables, both naturally supported by \ourname.
Infrequent multi-input/output operations can be decomposed into multiple stages.

\myparagraph{Vector Search}
Approximate nearest neighbor (ANN) is a key component in vector search.
Many state-of-the-art designs use bucket-based methods~\cite{tian2023db}, graph-based approaches~\cite{jayaram2019diskann,malkov2018efficient}, or hybrid schemes~\cite{munoz2019hierarchical,manohar2024parlayann} that maintain a nearest-neighbor graph at higher levels to guide searches toward candidate buckets for detailed vector comparisons.
Bucket traversal naturally fits the orchestration framework, while graph traversal resembles a vertex-centric graph algorithm, which can also be expressed using orchestration.

\myparagraph{Fault-Tolerant Storage}
Our orchestration framework can also support RAID~\cite{chen1994raid} or erasure coding~\cite{balaji2018erasure} in fault-tolerant storage, by setting \textsc{InputPointers} to fetch the original data, \textsc{wb} to write back to the target location at \textsc{OutputPointers}, and $f$ to perform the RAID or erasure-coding computation.

\section{\ourname: System Design}
\label{sec:main_design}

\textbf{\textit{\ourname}} is a practical orchestration framework that provides provable load-balanced scheduling for Fig.\ref{fig:orchestration_interface}'s API, with low communication overhead under contention.
It uses a (multi-)tree-structured routing scheme (\S\ref{subsubsec:comm_forest}) across machines that may resemble prior work~\cite{jelasity2005gossip,vadhiyar2000automatically}, but the task-data scheduling inside this routing is fundamentally different.

Its key idea is to judiciously \textit{exploit both directions} of flow: \textbf{\textit{push}}, where a task context is sent to the machine storing its required data chunk, and \textbf{\textit{pull}}, where the data chunk is fetched to the machine where the task resides.
Such bidirectional strategies (referred to as \textbf{\textit{distributed push-pull}}) are rarely emphasized in prior work, which typically adopts only one direction~\cite{vora2014aspire,li2020pegasus,gog2016firmament,fan2011small,murray2013naiad,jelasity2005gossip,vadhiyar2000automatically,rai2018load,tardieu2014glb,beltran2022ompss2cluster,montesano2024spatial}.
Only a few prior works show that this strategy can significantly improve load balance, but they are limited to PIM settings~\cite{kang2022pimtreevldb,kang2023pimtrie,zhao2025optimal,zhao2026pimzd,kang2025pim,kang2026nonclairvoyant}, relying on an extremely powerful centralized coordinator.
\ourname proposes a more widely applicable push-pull scheduling for distributed systems (\S\ref{subsubsec:dist_push_pull}) that does not require a powerful centralized coordinator.

\myparagraph{Execution Pipeline}
Aligned with the original three-step orchestration in \S\ref{subsec:abstraction_interface}, \ourname refines Step~(i) (\textit{Read Data}) by splitting it into two substeps---\textit{contention detection} and \textit{\taskdata co-location}.
This yields a four-phase execution pipeline:

\begin{enumerate}[label=(\arabic*),topsep=0pt,itemsep=0pt,parsep=0pt,leftmargin=15pt]
    \item \textbf{Contention Detection}: As a preprocessing phase during execution, we first determine which data chunks are hot and which machines are contending for the hot data chunks, as follows. 
    The information of each task is sent to the machine storing its target data to gather the contention information for each data chunk. To preserve load balance, messages are passed along a communication forest network (\S\ref{subsubsec:comm_forest}), in the form of meta-tasks (\S\ref{subsubsec:meta_task}).
    \item \textbf{\TaskData Co-location}: Using the information gathered in Phase~1, tasks and data chunks are redistributed across machines to achieve an optimized schedule that minimizes communication and computation time (\S\ref{subsubsec:dist_push_pull}).
    \item \textbf{Task Execution}: Now that each task and (a copy) of its target data are co-located on the same machine, the tasks are executed locally.
    \item \textbf{Data Write-backs}: The updated value for any data chunk modified by a task is written back to the stored location, with contending updates ``merged'' (\S\ref{subsubsec:write_back}).
\end{enumerate}
For iterative applications, each executed task in a stage dictates a task to be executed in the next stage, and this execution pipeline is repeated, until all iterations are done.

\subsection{Communication Forest for Handling Skew}
\label{subsubsec:comm_forest}

Although we assume a fully connected network, tasks and data cannot be co-located via direct push or direct pull, as hot data would cause severe load imbalance.
A circular challenge exists: \textit{we need a preprocessing Phase~1 to gather contention information for later load balancing, but Phase~1 itself must stay load-balanced before contention is even known}.

Inspired by~\cite{jelasity2005gossip,vadhiyar2000automatically}, in Phase~1, rather than exchanging all \taskdata in a single communication round, a tree-structured network in Fig.\ref{fig:counting_tree} is used to gradually push the information of tasks (which reside in the leaf machines and need to access a data chunk in the root machine $r$ of the tree) toward their target data, moving up one tree level per round.
This prevents tasks accessing the same hot data chunk from overloading machine $r$.
Formally, a \textbf{\textit{communication tree}} rooted at machine $i$ is a network with a balanced tree structure with $P$ leaves and fanout $F$, as shown in Fig.\ref{fig:counting_tree}.
Each of the $P$ leaf nodes represents a different physical machine, while internal nodes correspond to virtual transit machines, each of which is mapped to one of the $P$ machines via a hash function globally known to all machines.
Notably, a physical machine can simultaneously serve multiple roles: leaf, internal node, root.

The communication tree collects information about all tasks that need to access data stored on machine~$i$.
It enables efficient collection of task-related information (such as reference counts) with load-balanced communication.
\S\ref{subsubsec:meta_task} pre- sents the format of information passed along this network.

\begin{figure}[t]
    \centering
    \includegraphics[width=0.8\linewidth]{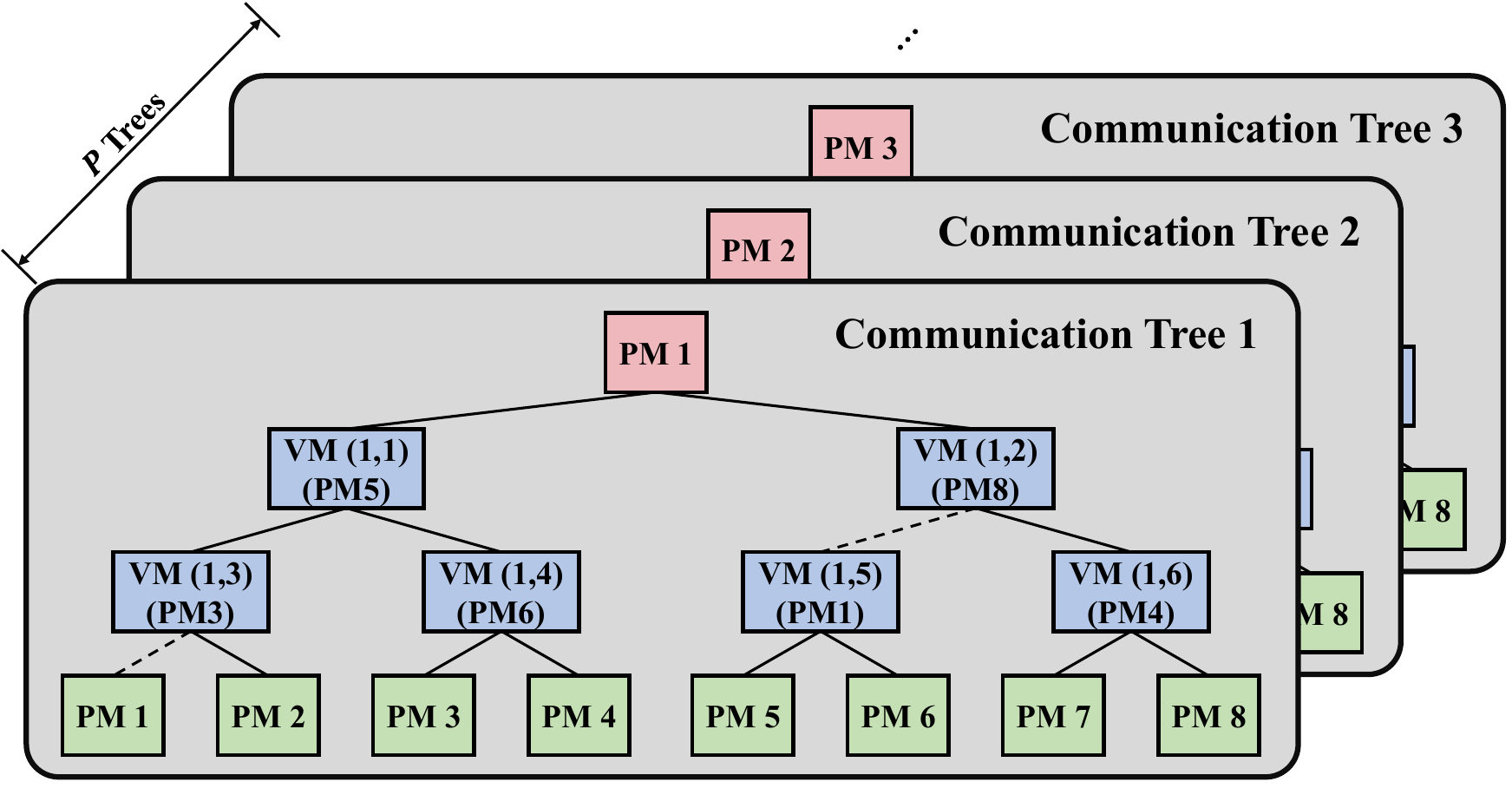}
    \vspace{-0.7em}
    \caption{A \textit{communication forest} with fanout $F=2$ for $P=8$ machines. Green nodes are source machines (w/tasks), red nodes are destination machines (w/data), and blue nodes are intermediate transit virtual machines (VMs) with IDs $(\text{root}, \text{BFS\_id})$, mapped to physical machines (PMs) via example hash function $h(x,y)=(x+3y)\bmod8 + 1$. Dashed edges mean that PM does not need to send messages to itself.}
    \vspace{-1em}
\label{fig:counting_tree}
\end{figure}

In the Fig.\ref{fig:counting_tree} example, Phase~1 needs three rounds (the height of trees).
Suppose machines~2 and~3 each have a task accessing data on machine~1.
In round~1, each task sends a packed message (data address and task info) to its respective parents, VM(1,3)=PM3 and VM(1,4)=PM6.
In round~2, messages are forwarded to VM(1,1)=PM5.
In the final round, an aggregated message is delivered to machine~1.
This message aggregation prevents the overloading of a single machine.

A \textbf{\textit{communication forest}} consists of $P$ such communication trees, each rooted at a distinct machine.
All tasks send messages in parallel along the trees associated with the machines that store their target data chunks.
Although the internal transit machines are statically chosen within the communication forest, each data chunk is stored on a random machine (\S\ref{subsec:abstraction_system_model}).  
This randomization in the corresponding tree selection is equivalent to dynamically selecting transit nodes in both theoretical analysis and practical implementation.

We set $F = \Theta\left(\log P / \log \log P\right)$ using theoretical analysis (\S\ref{subsec:orch_formal_guarantee}),
which also has strong practical performance (\S\ref{sec:key_value_case_study}, \S\ref{sec:eval}).

\begin{figure}[t]
    \centering
    \includegraphics[width=0.88\linewidth]{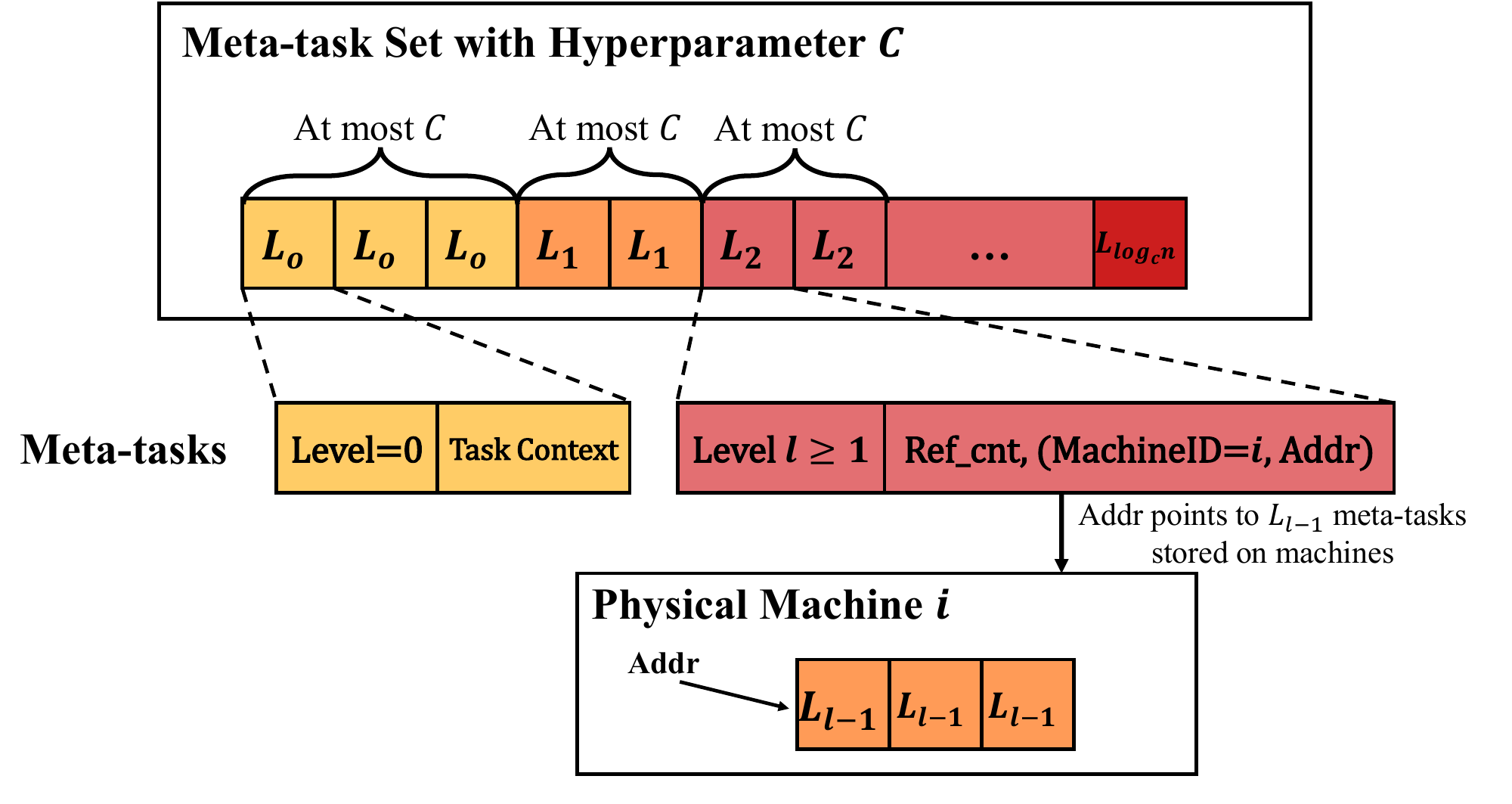}
    \vspace{-1.2em}
    \caption{A \textit{meta-task set}, zooming in on $L_0$ and $L_l$ ($l\geq 1$) meta-tasks. Deeper colors indicate higher levels. $L_0$ contains original task contexts, while for $l \geq 1$, $L_l$ contains a pointer to an $L_{l-1}$-meta-task array stored on a machine.}
    \vspace{-1.0em}
\label{fig:meta_task}
\end{figure}

\subsection{Meta-Task Structure}
\label{subsubsec:meta_task}

Before \ourname performs orchestration of tasks and data in Phase~2, their information must first be exchanged via message passing.
These messages support co-location by carrying \romenum{1} the reference count of each data chunk and \romenum{2} the locations of the tasks requesting it.
They should also be \textit{aggregatable} to prevent unbounded growth in message sizes.

To meet these goals, we introduce a \textbf{\textit{meta-task}} structure---parameterized by $C$ and $L$---to serve as messages in Phase~1 within the tree-structured network.
At a high level, if the reference count of a data chunk is $\le C$, the meta-task contains the full contexts of all requesting tasks.
If the reference count exceeds $C$, the meta-task includes only aggregated metadata, including the count of tasks and a pointer to retrieve their locations,
thereby preventing unbounded message growth.

\begin{figure}[t]
    \centering
    \includegraphics[width=\linewidth]{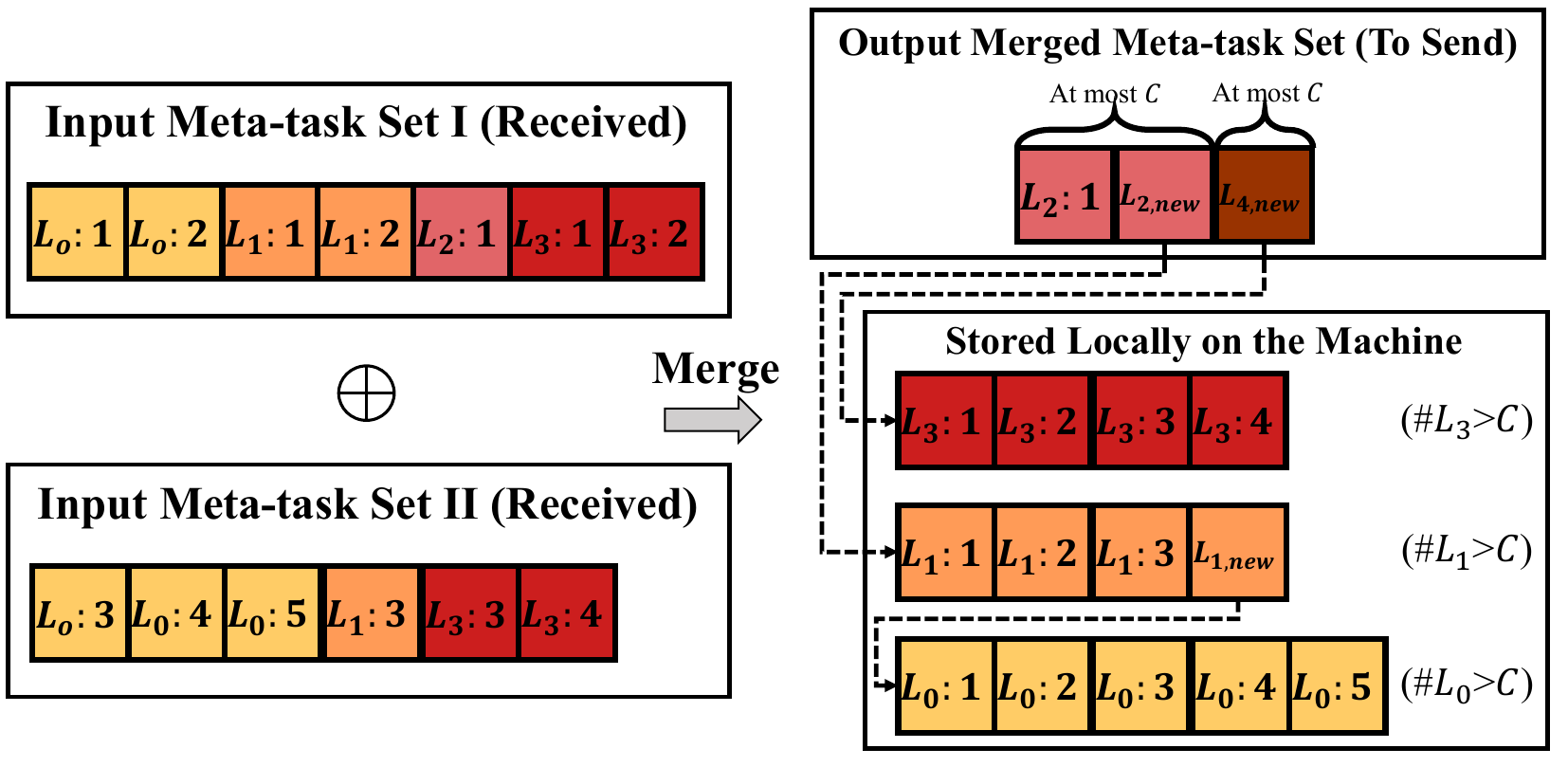}
    \vspace{-2.2em}
    \caption{\textit{Merging} two meta-task sets with $C=3$. Meta-tasks on lowest level $L_0$ are merged first and generate $L_{1,new}$. Then $L_1$-meta-tasks (including both the received $L_1$ and the newly generated $L_{1,new}$) are merged and generate $L_{2,new}$. Continue the process until the top level of $L_4$.}
    \vspace{-1.2em}
\label{fig:merge_meta_task}
\end{figure}

The attribute $L$, which serves as a level indicator, enables hierarchical aggregation of meta-tasks.
Raw task contexts are treated as $L_0$ meta-tasks.
During message passing in the communication forest, if an intermediate machine receives more than $C$ $L_l$ meta-tasks, it locally stores them and creates a new $L_{l+1}$ meta-task.
This new $L_{l+1}$ meta-task contains two pieces of information:  
\romenum{1} the aggregated reference count, and  
\romenum{2} the (machine ID, address) pair indicating where the $C$ underlying $L_l$ meta-tasks are stored.
Only this aggregated $L_{l+1}$ meta-task will be forwarded to its parent machine.

A \textbf{\textit{meta-task set}}, as Fig.\ref{fig:meta_task} shows, is a set structure containing multiple meta-tasks.
The number of meta-tasks in a meta-task set with the same $L$ value should not exceed $C$.

A \textit{merge} of two meta-task sets works like adding two base-$C$ integers: the meta-tasks at level $L_l$ are the digit in the $C^l$ place, at most $C$ per level.
The merge unions all meta-tasks from both sets, then scans from the lowest level up.
Whenever a level $L_l$ exceeds $C$ it \textit{carries}---those meta-tasks are stored locally and replaced by one $L_{l+1}$-meta-task pointing to them.
Carries propagate up until every level holds at most $C$.
For example in Fig.\ref{fig:merge_meta_task}, $L_0$ overflows and carries $L_{1,\text{new}}$. 
$L_1$ level, now including $L_{1,\text{new}}$, overflows and carries $L_{2,\text{new}}$.
The $\le C$ $L_2$-meta-tasks are unchanged.
Continue to the top.

This merging process guarantees that each meta-task set has bounded size $C\log_C n$, while containing sufficient information to enable efficient \taskdata exchange in Phase~2.

We use $C=\Theta(B/\sigma)$ in \ourname following the theory-guided design in \S\ref{subsec:orch_formal_guarantee}, which also leads to strong practical performance as shown in \S\ref{sec:key_value_case_study} and \S\ref{sec:eval}.

\subsection{Distributed Push-Pull for Co-location}
\label{subsubsec:dist_push_pull}

After Phase~1, each data chunk receives a meta-task set that contains all information needed for \taskdata co-location.
In Phase~2, we introduce \textbf{\textit{distributed push-pull}}, a strategy that exploits the bidirectional flow of tasks and data to resolve potential contentions on data chunks.

\myparagraph{Push}
If the reference count of a data chunk is $\le{C}$, then all task contexts requesting this chunk have already been delivered as $L_0$ meta-tasks to the machine storing the chunk, completing Phase~1 for these lightly contended tasks.

\myparagraph{Pull}
For contended tasks where the reference count of a data chunk exceeds $C$, tasks are stored on intermediate transit machines in the communication forest.
A \textit{pull} process---sending data (copies) to tasks---is required in Phase~2.  
Location information in meta-task sets guides the pull.

As an example, suppose a data chunk receives a meta-task set with the highest level $L_l$.  
It first broadcasts a copy of the chunk to all $\Theta(C)$ machines storing the $L_{l-1}$-meta-tasks, then chases pointers and broadcasts the data to the $\Theta(C^2)$ machines storing the $L_{l-2}$-meta-tasks, continuing until all $L_0$ meta-tasks (the original task contexts) have received a copy.

This broadcast goes through an inherent tree-structured network formed by a meta-task set, which we define as the \textbf{\textit{meta-task tree}}.
This tree is more compact than the communication tree, forming a full $C$-ary tree structure with all leaf machines storing $L_0$-meta-tasks and no empty leaves.

\myparagraph{Key Takeaway}
Bi-directional distributed push-pull
combines push and pull: lightly contended tasks are directly pushed to their target data, whereas more contended tasks pull their data gradually along the created meta-task tree.

\subsection{Write Backs}
\label{subsubsec:write_back}
If tasks require write-backs that are \textit{merge-able} (defined below), they can be efficiently aggregated, allowing the write-backs from multiple tasks to the same data chunk to be combined.
In Phase~4, these write-backs are propagated along a similar tree-structured path over the \textsc{OutputPointer}, with aggregation performed during message passing.

\vspace{-0.1em}
\begin{dfn}[Merge-able Operations]
\label{dfn:mergeable}
    An operation $\oplus$ is \textbf{merge-able} if there exist operations $\odot$ and $\otimes$ such that for any data $x$ and $y_i(1\leq i\leq n)$, $x\oplus y_1 \oplus ... \oplus y_n=x\odot(y_1 \otimes ... \otimes y_n)$.
\end{dfn}
\vspace{-0.1em}

Merge-able operations $\oplus$ are commonly seen, including and not limited to \romenum{1} idempotent operations ($\odot$ is $\oplus$ and $\otimes$ keeps one argument), \romenum{2} set-associative operations ($\odot$ and $\otimes$ are both $\oplus$), \romenum{3} concurrent writes where only one random write operation succeeds in a concurrent batch ($\odot$ is \textsc{write} and $\otimes$ is \textsc{random select}), and \romenum{4} write operations where concurrent writes are resolved using a deterministic decision process, e.g., within a batch only the one with the largest timestamp / transaction ID succeeds ($\odot$ is \textsc{write} and $\otimes$ is finding the value to write using the decision process).

\subsection{Theoretical Guarantees for \ourname}
\label{subsec:orch_formal_guarantee}

\Cref{theorem:main} proves \ourname to be \textbf{\textit{asymptotically optimal}} compared with the lower bounds provided in~\cite{roughgarden2018suffles}.
Please refer to \suppref{appendix:main_proof} for its proof.

\vspace{-0.2em}
\begin{theorem}[Main Theorem]
\label{theorem:main}
    Assume $n=\Omega(\frac{PB\log^{3}{P}}{\log\log{P}})$.
    If each machine holds $\Theta(n/P)$ tasks when a stage begins, and write backs are merge-able, then
    \romenum{1} \ourname takes $O(\frac{n}{P}(\min\{B,$ $\sigma\}+\log_{\frac{n}{P}}{P}))$ communication time \whp\footnote{We use $O(f(n))$ \textbf{with high probability} (\textbf{\whp}) (in $n$) to mean $O(cf(n))$ with probability at least $1-n^{-c}$ for $c\geq{1}$.} in this stage;
    and \romenum{2} each machine executes $\Theta(\frac{n}{P})$ tasks \whp.
\end{theorem}
\vspace{-0.2em}

\myparagraph{Key Takeaways}
\romenum{1} \textbf{\textit{Optimal Communication}}: TD-Orch incurs optimally minimal communication time as $P$ scales up. 
\romenum{2} \textbf{\textit{Optimal Computation}}: TD-Orch balances computation work, maximizing system throughput.
\romenum{3} \textbf{\textit{Inductive Execution}}: Tasks remain load-balanced across machines after each stage, enabling efficient execution of subsequent stages.

\section{Case Study I: \KV Store}
\label{sec:key_value_case_study}

\begin{figure}
    \centering
    \includegraphics[width=\linewidth]{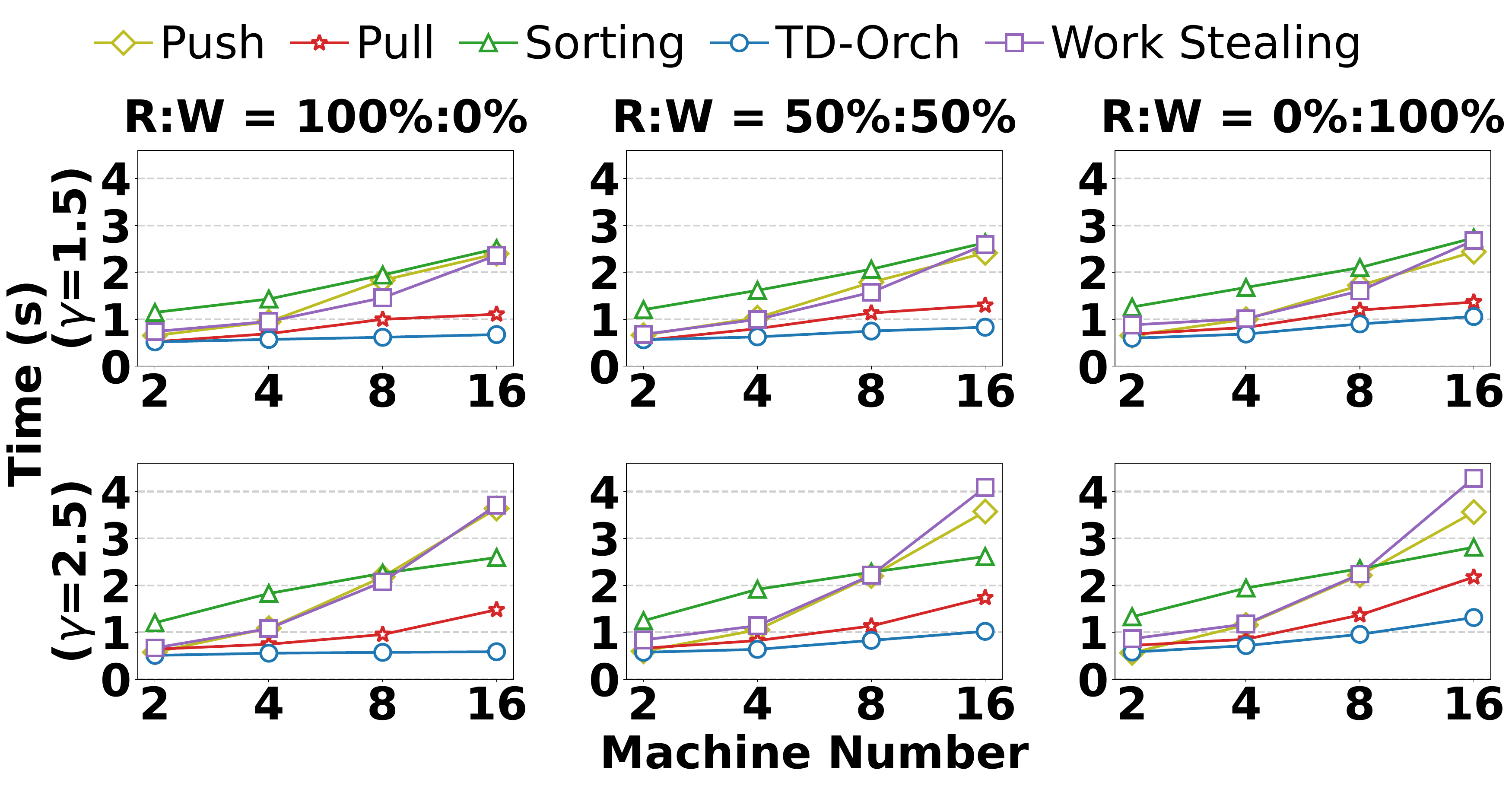}
    \vspace{-2.2em}
    \caption{Runtime results (in seconds) for YCSB workloads A (50\% reads, 50\% writes), C (read-only), and LOAD (write-only). \ourname scales well and outperforms prior methods. Results for YCSB-B (95\% reads, 5\% writes) and $\gamma$=1.6--2.4 exhibit similar trends and are omitted due to space constraints.}
    \vspace{-0.8em}
    \label{fig:gamma_linecharts}
\end{figure}

A natural fit for our orchestration framework is a \kv store, which we implement atop TD-Orch using a concurrent distributed hash table as the underlying storage structure.

We compare \ourname with Direct Push, Direct Pull, Sorting-based (described in \S\ref{subsec:preliminary_prior_work}), and a decentralized work stealing load balancer~\cite{blumofe1999scheduling}, all using MPI for inter-machine communication.
Direct Push, Direct Pull, Work Stealing, and \ourname are implemented directly, while Sorting is built on KaDiS~\cite{axtmann2015practical,axtmann2018lightweight,axtmann2017place,harsh2019histogram}, a state-of-the-art MPC sorting library.
In work stealing, each machine initially owns the tasks for its keys.
Tasks on overloaded machines are dynamically migrated to idle machines, with the required data fetched via a Direct-Pull path.
All methods follow the interface in Fig.\ref{fig:orchestration_interface}.

\myparagraph{Platform Details}
We use a cluster of 1,024 cores---16 machines each equipped with AMD Opteron 6272 processors (64 cores, 2.1 GHz), 64 MB L3 cache, and 128 GB memory.
Machines are interconnected via a 10~Gbps network.

\myparagraph{Workloads}
We evaluate the orchestration methods on YCSB workloads~\cite{cooper2010benchmarking} A (50\% reads, 50\% writes), B (95\% reads, 5\% writes), C (read-only), and LOAD (write-only), with data access patterns following Zipf distribution~\cite{zipf2016human}, a common setting in database/storage benchmarks (e.g.,~\cite{kang2022pimtreevldb,luo2023smart}).
Each task fetches a data item, performs a multiply-and-add operation, and then writes the value back if updated.

\myparagraph{Execution Time}
Fig.\ref{fig:gamma_linecharts} shows the results of weak-scaling experiments with $P=2$--$16$ machines.
In this setup, each machine is initialized with 2~million tasks in each batch, with target data addresses drawn from Zipf exponents $\gamma=1.5$--$ 2.5$.

\ourname is \romenum{1} resilient to high skew and delivers near-linear performance under weak scaling, and \romenum{2} significantly faster than other methods.
Across 2$\sim$16 machines, \ourname achieves geometrically-averaged speedups of $2.09\times$, $1.42\times$, $2.83\times$ and $2.34\times$ over the other four baselines, respectively.

\myparagraph{Load Balance}
In the 16-machine experiment for $\gamma=2.5$, \ourname achieves a coefficient of variation across machines of $0.09$, compared with $0.29$/$0.15$ for Direct-Push/Direct-Pull, respectively.
Work stealing achieves a coefficient of $0.27$, because every hot-key access is still fetched from the data owner, so dynamically migrating tasks alone cannot relieve data contention.
Although Sorting achieves the same value of $0.09$, it is substantially slower in terms of absolute runtime.

\myparagraph{Ablations}
We conduct three additional ablations on \ourname.
Full results are presented in \suppref{sec:appendix_kv_ablation}.
\romenum{1} Component Breakdown: A runtime breakdown of each phase in \ourname is presented.
\romenum{2} Memory Overhead: \ourname{'s} tree-structured metadata adds only a small peak-memory overhead over Direct Pull (avg $9.7\%$, up to $15.5\%$).
\romenum{3} Hyperparameter Sensitivity:
When $F$ and $C$ are varied within the same \textit{asymptotic regime}, runtime changes by less than $20\%$, showing that the theoretical guarantees in \S\ref{subsec:orch_formal_guarantee} translate into \textit{stable} practical performance.

\section{Case Study II: Graph Processing Systems}
\label{sec:graph_processing_system}

As a more extensive case study, we implement a distributed graph processing system named \textbf{\ourgraph} on top of our \ourname framework to show the effectiveness of \ourname.
Following common practices in prior vertex-centric works~\cite{zhu2016gemini,wu2021seastar,wang2023scaleg,zhu2022vctune}, \ourgraph maintains two primary storage structures: one for vertices, which hold their intermediate values and are pinned to specific machines based on a predefined mapping during graph ingestion; and another for edges and weights, which is dynamically orchestrated by the system.

\myparagraph{Interface}
Our graph processing engine is primarily executed through a combined interface of \ourname (Fig.\ref{fig:orchestration_interface}) and the \EDGEMAP abstraction.
Ligra~\cite{ShunB2013,shun2015ligraplus,dhulipala2017julienne} and GBBS~\cite{dhulipala18scalable} introduce \EDGEMAP, which works on a single machine and takes as input a graph $G$, a \VERTEXSUBSET $U$, and two boolean functions, $F_1$ and $F_2$.
It applies $F_1$ to each edge $(u,v)$ iff $u\in U$ and $F_2(v)=true$, returning a new \VERTEXSUBSET $U^{\prime}$ consisting of all vertices $v$ for which $F_1(u,v)$ is applied and returns true.
\EDGEMAP supports two execution modes—sparse and dense—ensuring that it runs in $O(\sum_{u\in U} \text{deg}(u))$ work (total computation) and $O(\log n)$ depth for $n$-node graphs.

We refer to our new interface as \textbf{\DISTEDGEMAP}, which is defined in Fig.\ref{fig:distedgemap}.
We treat $G$ as directed, representing each undirected edge $\{u,v\}$ as two directed edges $(u,v)$ and $(v,u)$.
$U$ is a \DISTVERTEXSUBSET, which stores on each machine a \VERTEXSUBSET holding the active vertices that machine owns.
An execution function $f$ is applied to all edges $(u,v)\in E$ where $u \in U$, and returns a value.
All returned values are aggregated by \DISTEDGEMAP using a user-defined merge-able operator \textit{merge\_value}, and the result is \textit{write\_back} to $v$.
\DISTEDGEMAP outputs a \DISTVERTEXSUBSET comprised of the vertices whose \textit{write\_back} returns true.

As Fig.\ref{fig:distedgemap} shows, \DISTEDGEMAP deliberately mirrors the single-machine \EDGEMAP, hiding distributed architecture and scheduling from the user.

\begin{figure}
\begin{lstlisting}
function Dist_Edge_Map(
    G=@$(V,E)$@: Directed Graph,  
    U: DistVertexSubset,  // Current iteration
    f: function,  // Execution
    write_back: function,
    merge_value: function  // Aggregate wb
) -> DistVertexSubset;  // For next iteration
\end{lstlisting}
\vspace{-1.5em}
\caption{\DISTEDGEMAP interface}
\vspace{-1.2em}
\label{fig:distedgemap}
\end{figure}

\myparagraph{Example Graph Algorithms}
We implement five representative graph algorithms using \DISTEDGEMAP in our system: breadth-first search (\textbf{BFS}), single-source shortest path (\textbf{SSSP}), betweenness centrality (\textbf{BC}), connected components (\textbf{CC}), and PageRank iteration (\textbf{PR}).
In all of them, users define only $f$ and \textit{write\_back}, then launch \DISTEDGEMAP.
Please refer to \suppref{sec:appendix_bc_alg} for an implementation of a complex algorithm such as BC in \textit{\textbf{fewer than 70 lines of code}} using our interface, which is much more concise than prior works (e.g., 400+ lines in Gemini~\cite{zhu2016gemini}).

\subsection{Main Design of \ourgraph}
\label{subsec:overall_design_outline}

As outlined in \S\ref{subsec:abstraction_application}, \DISTEDGEMAP instantiates Fig.\ref{fig:orchestration_interface} with one task per edge---\textit{edges correspond to tasks, and vertex values correspond to data}.
The task for edge $(u,v)$ carries the edge as its \textsc{LocalContexts}, the values of $u$ and $v$ as its \textsc{InputPointer}/\textsc{OutputPointer}, and $f$ as its lambda.
Its \textit{write\_back} operation, together with \textit{merge\_value}, implements the mergeable write-back described in \S\ref{subsubsec:write_back}.
A key observation is that contention for a vertex value from its incident edge tasks is proportional to the vertex degree, which is not/slowly changing unless the graph is substantially updated.

Based on this observation, \ourname can be invoked \textbf{\textit{periodically}} every few stages, to mitigate skew during the \DISTEDGEMAP in between.
Subsequent \DISTEDGEMAP stages reuse these pre-orchestrated flows instead of re-orchestrating dynamically at every stage.
This periodic load balancing is especially beneficial for highly skewed graphs, such as the power-law graphs~\cite{gonzalez2012powergraph} commonly encountered in practice.

This periodic process consists of two stages.
In the first stage, we begin by randomly placing each edge on an arbitrary machine and then execute one round of \ourname, which attempts to colocate all edges with their source vertices.
Edges whose sources are low-degree vertices are now placed on the same machine with the vertices.
For edges with high-degree sources, the meta-task tree from \S\ref{subsubsec:dist_push_pull}, referred to here as a \textit{source tree}, captures the source-side communication paths.
We then fix the storage location of these edges and propagate future source-vertex values through the source trees.
In the second stage, another round of \ourname is performed, where each edge, starting from its location at the end of the first stage, tries to reach its destination vertex.
This similarly constructs a meta-task tree, denoted \textit{destination tree}, along which destination-vertex values are written back using $write\_back$.
Subsequent iterations reuse these established flows without re-orchestration.
Phases~1--2 are thus amortized, and each iteration executes only Phase~3 (where edge tasks run where their edges are stored) and Phase~4 (where results are written back along the trees).

Adopting this periodic processing provides two key benefits:  
\romenum{1} it organizes edge storage into chunks (possibly on transit machines, rather than on source/destination machines, if highly contended), ensuring balanced load; and  
\romenum{2} it establishes a tree-shaped communication framework for both value fetch and \textit{write\_back}, enabling aggregation communication.
This design differs from prior ghost/mirror node designs (e.g.,~\cite{zhu2016gemini,gonzalez2012powergraph}), which neither use transit machines to mitigate contention nor employ tree-structured networks to improve communication costs.  
Consequently, our design achieves better load balance and improved performance (\S\ref{sec:eval}).

\myparagraph{Sparse-Dense Execution}
We design a dual-mode execution flow for distributed graph processing that can switch between \textit{sparse} and \textit{dense} modes to efficiently exploit the data layouts provided by \ourname.
The mode is determined by the size of the current \DISTVERTEXSUBSET.
Specifically, the \textit{sparse} mode is optimized for small \DISTVERTEXSUBSET, while the \textit{dense} mode is optimized for large ones.

In the \textit{sparse} mode, a vertex-centric execution flow is used: each active vertex $u \in U$ propagates its value down the source tree, activating all corresponding edges.
The result values are then sent and merged along the destination tree.

In contrast, the \textit{dense} mode uses an edge-centric execution flow.
All active vertices $u \in U$ broadcast their values to all machines.
Each edge $(u,v)$ is then executed if $u \in U$.

The choice between sparse and dense mode is dynamically determined by the size of the active frontier in the current round.
Specifically, if $\Sigma_{u\in{U}}\deg(u)=o(P\cdot|U|)$ where $d(\cdot)$ denotes the out-degree of inter-machine edges, the sparse mode is used; otherwise, the dense mode is used.
This dynamic switching ensures that the execution remains both work- and communication-efficient in theory (with work/communication of $O(\Sigma_{u\in{U}}\deg(u))$ as in Ligra and load-balanced due to \ourname), while also achieving high practical performance.

\subsection{Graph-Specific Implementation Techniques}
\label{subsec:implementation_technique_outline}

\ourgraph applies multiple additional techniques orthogonal to \ourname, which can be grouped into three categories.
Due to space limit, we summarize here their key functionalities.
Details are given in \suppref{appendix:implementation_techniques}.

\myparagraph{(T1) Optimized Communication}
We use deduplication in all messages, and destination-vertex-based rank filtering in broadcast to further reduce inter-machine communication.

\myparagraph{(T2) Work-Efficient Local Computation}
Our design leverages ParlayLib~\cite{blelloch2020parlaylib}, a state-of-the-art parallel library, to implement work-efficient computations.
We improve \VERTEXSUBSET with a concurrent hash table + bitmap, and further develop a lightweight local \EDGEMAP.
\DISTEDGEMAP also accepts an optional \textit{filter\_dst} function, excluding vertices that are guaranteed not to be destinations in current round.

\begin{table}[t]
\small
\centering
\caption{Graph Algorithms Bounds ($n$ nodes, $m$ edges).}
\label{tab:graph_bounds}
\vspace{-0.75em}
\begin{tabularx}{3.33in}{@{}c@{\ }ccccc@{}}
\toprule
\textbf{Algorithm} & BFS & SSSP & BC & CC & PR \\
\midrule
Computation & $O\left(\frac{m}{P}\right)$ & $O\left(\frac{m}{P}\right)$ & $O\left(\frac{m}{P}\right)$ & $O\left(\frac{m}{P}\cdot{\text{diam}}\right)$ & $O\left(\frac{m}{P}\right)$ \\
Communication & \multicolumn{5}{c@{}}{All equal to $O(\text{Computation}\cdot{\log_{\frac{n}{P}}{P}})$} \\
\bottomrule
\end{tabularx}
\vspace{-1.0em}
\end{table}

\myparagraph{(T3) Aligned Coordination between Components}
We design tighter coordination between system components---between threads and scheduler, among threads to prevent cache thrashing, and between local and global execution.

\subsection{Theoretical Guarantees for \ourgraph}

Since \Cref{theorem:main} applies to \DISTEDGEMAP, all of our algorithms are work-efficient and load-balanced across machines.
The theoretical bounds summarized in Tab.\ref{tab:graph_bounds} are achieved when $n=\Omega(P\log^{2}{P})$.
While most prior works~\cite{zhu2016gemini,mofrad2020graphite,ahmad2018la3,yan2015effective,lu2014large,yan2014pregel,wang2021gnnadvisor,serafini2021scalable} focus on empirical system improvements, we place strong emphasis on graph-theoretic guarantees as guidance for implementation.
In \S\ref{subsec:e2e_results}, these guarantees will translate into practical performance improvements.

\subsection{Prior Distributed Graph Processing Systems}
Classic distributed frameworks such as Pregel~\cite{malewicz2010pregel,bu2014pregelix,han2015giraph} and MapReduce~\cite{dean2008mapreduce} inspired early graph systems,
such as
GraphLab~\cite{low2012distributed},
PowerGraph~\cite{gonzalez2012powergraph},
GraphX~\cite{gonzalez2014graphx},
and more~\cite{bulucc2011combinatorial,sundaram2015graphmat,gadepally2015graphulo,khayyat2013mizan,yuan2014fast,kang2009pegasus,xie2015sync,venkataraman2013presto,edmonds2013expressing}.
We classify more recent graph systems into two families:
\romenum{1} \textbf{\textit{Graph-algorithm-based}}:
Gemini~\cite{zhu2016gemini},
Gluon~\cite{dathathri2018gluon},
PowerLyra~\cite{chen2019powerlyra},
and more~\cite{lee2015scaling,hong2015pgxd,lin2018shentu,li2018regraph,dathathri2019phoenix,li2019topox,liu2024faasgraph,wang2016gunrock,pai2016compiler,liu2015enterprise,pandey2021terrace,zheng2015flashgraph,yin2023glign,wang2018lazygraph,kusum2016efficient,li2023liberator,liu2023mbfgraph,niu2026diggerbees,zhang2017finepar,yu2021dfograph,gan2024graphcube,harshvardhan2015hierarchical,sun2025helios,wang2026elasgnn,vora2017coral,firoz2018runtime,cheramangalath2017dhfalcon,newaz2025locality,cap2022scaling}; and
\romenum{2} \textbf{\textit{Linear-algebra-based}}:
GraphScope~\cite{fan2021graphscope},
GraphPad~\cite{anderson2016graphpad}, BLAS~\cite{azad2022combinatorial},
and more~\cite{mofrad2020graphite,mofrad2019efficient,ahmad2018la3,liu2015csr5,hong2019tiling,li2013smat,wang2013augem,gianinazzi2024arrow}.

Compared with prior graph-algorithm-based systems, our \ourgraph achieves stronger theoretical bounds and higher practical performance, by better exploiting graph \textit{sparsity} and implementing genuine \textit{work-efficient} graph algorithms.

Compared with linear-algebra-based designs,  our \ourgraph, achieves stronger provable and empirical \textit{load balance} by using transit machines in its underlying \ourname to migrate contention.
In contrast, prior LA-based approaches alleviate only limited contention through empirical designs.

\section{Evaluation of Graph Processing Systems}
\label{sec:eval}

\subsection{Experimental Setup}
\label{subsec:experimental_setup}

\myparagraph{Platform}
We use the same 1024-core cluster as in \S\ref{sec:key_value_case_study}.

\myparagraph{Competitors}
Due to their state-of-the-art performance, we select \textit{GeminiGraph}~\cite{zhu2016gemini}, \textit{Graphite}~\cite{mofrad2020graphite}, and \textit{LA3}~\cite{ahmad2018la3} as baselines.
\textit{GeminiGraph} is the fastest graph-algorithm-based system, and the latter two are the fastest linear-algebra-based systems, representing the strongest alternatives in both design families (including being faster than~\cite{gonzalez2012powergraph,gonzalez2014graphx,chen2019powerlyra,anderson2016graphpad,apache2017giraph,low2012distributed}).
In addition, we include StarPlat~\cite{behera2024starplat} and Pregel+~\cite{yan2015effective,lu2014large,yan2014pregel} in one end-to-end experiment (\S\ref{subsec:e2e_results}), and the state-of-the-art single-machine design GBBS~\cite{dhulipala18scalable} as ablation (\S\ref{subsec:ablation_study}).

\myparagraph{Datasets}
We evaluate on Reddit~\cite{hamilton2017inductive}, uk-2005~\cite{BoVWFI,BRSLLP}, Twitter-2010~\cite{kwak2010twitter}, Friendster~\cite{yang2012defining}, Hyperlink-2012~\cite{meusel2015graph}, and Road-USA~\cite{road}, covering a wide range of dataset sizes and graph characteristics.
See the first column of Tab.\ref{tab:end_to_end_results}.

{
\footnotesize
\begin{table}[t]
\captionof{table}{End-to-end runtime (in seconds, fastest underlined).
}
\label{tab:end_to_end_results}
\vspace{-0.5em}
        \resizebox{0.97\linewidth}{!}{
        \begin{tabular}{@{} l l r r r r @{}}
        \toprule
        \textbf{Dataset} & \textbf{Alg.} & \textbf{\ourgraph} & \textbf{Gemini} & \textbf{Graphite} & \textbf{LA3} \\
        \midrule
         Reddit& BFS  & \underline{\textbf{0.015}} & 0.151 & 0.130 & 0.157 \\
         4 machines& SSSP & \underline{\textbf{0.071}} & 0.172 & 0.141 & 0.505 \\
         $n=2.33$M& BC   & \underline{\textbf{0.076}} & 0.275 & 0.507 & 0.437 \\
         $m=114$M& CC   & \underline{\textbf{0.016}} & 0.211 & 0.266 & 0.663 \\
         diam\footnotemark[2] $=6^{*}$ & PR   & \underline{\textbf{0.143}} & 0.676 & 0.297 & 1.92 \\
        
        \midrule
         uk-2005 & BFS  & \underline{\textbf{0.759}} & 8.98 & 8.01 & 32.56 \\
         8 machines& SSSP & \underline{\textbf{0.820}} & 10.68 & 13.79 & 45.41 \\
         $n=39.5$M& BC   & \underline{\textbf{3.00}} & 20.23 & 34.79 & 75.90 \\
         $m=482$M& CC   & \underline{\textbf{3.30}} & 11.22 & 20.85 & 63.11 \\
         diam $=276^{*}$& PR   & \underline{\textbf{2.37}} & 2.77 & 7.69 & 9.06 \\
        \midrule
         Twitter2010& BFS  & \underline{\textbf{0.938}} & 1.35 & 3.13 & 7.67 \\
         8 machines& SSSP & \underline{\textbf{2.05}} & 3.63 & 3.36 & 7.94 \\
         $n=41.7$M& BC   & \underline{\textbf{3.04}} & 4.34 & 6.16 & 15.54 \\
         $m=1.47$B& CC   & \underline{\textbf{1.28}} & 7.65 & 7.40 & 21.73 \\
         diam $=23^{*}$& PR   & 10.46 & 12.72 & \underline{\textbf{9.85}} & 13.67 \\
        \midrule
         friendster& BFS  & \underline{\textbf{1.98}} & 3.12 & 3.70 & 12.38 \\
         8 machines& SSSP & \underline{\textbf{3.09}} & 4.42 & 4.25 & 27.23 \\
         $n=65.6$M& BC   & \underline{\textbf{5.29}} & 8.24 & 7.81 & 69.15 \\
         $m=1.80$B& CC   & \underline{\textbf{2.52}} & 17.92 & 9.92 & 56.25 \\
         diam $=32$& PR   & \underline{\textbf{18.32}} & 33.07 & 21.71 & 28.03 \\
        
        \midrule
         Hyperlink12& BFS  & \underline{\textbf{1.60}} & 6.22 & 5.59 & 20.50 \\
         16 machines& SSSP & \underline{\textbf{2.40}} & 8.02 & 28.53 & 33.10 \\
         $n=102$M& BC   & \underline{\textbf{3.03}} & 16.09 & 46.15 & 51.50 \\
         $m=0.93$B& CC   & \underline{\textbf{5.26}} & 11.71 & 29.42 & 48.31 \\
         diam $=95^{*}$& PR   & \underline{\textbf{7.12}} & 10.10 & 14.43 & 38.78 \\
         \midrule
        Road-USA & BFS  & \underline{\textbf{10.68}} & 357.73 & 1224.83 & 1431.11 \\
        16 machines & SSSP & \underline{\textbf{15.76}} & 362.11 & 1651.52 & 1788.53 \\
        $n=23.9$M & BC   & \underline{\textbf{26.24}} & 763.09 & {$>$35min} & {$>$35min} \\
        $m=28.9$M & CC   & \underline{\textbf{12.93}} & 383.51 & 1814.93 & {$>$35min} \\
        diam $=6139^{*}$ & PR   & 1.00 & \underline{\textbf{0.612}} & 3.01 & 3.59 \\
        \bottomrule
         \end{tabular}
         }
         \vspace{-1.5em}
\end{table}
\footnotetext[2]{An asterisk (*) denotes an estimated diameter obtained using approximation, following the common practice in Ligra~\cite{ShunB2013} and GBBS~\cite{dhulipala18scalable}.}
}

\subsection{End-to-End Results}
\label{subsec:e2e_results}

Tab.\ref{tab:end_to_end_results} presents a performance comparison between \ourgraph and our three state-of-the-art baselines, evaluated on various graph datasets.
Overall, \ourgraph consistently achieves superior performance, outperforming the baselines in \textbf{\textit{28 out of 30}} tested scenarios, most by substantial margins.

\textbf{\textit{\ourgraph}} achieves speedups more than $2\times$ in $19$ out of the $30$ tests we carry out.
These improvements arise from two key factors: \romenum{1} \ourgraph more effectively exploits sparsity through sparse-dense execution in \S\ref{subsec:overall_design_outline} and the techniques in \S\ref{subsec:implementation_technique_outline}; and \romenum{2} it reduces inter-machine communication while maintaining better load balance via periodic \ourname in \S\ref{subsec:overall_design_outline}.
These advantages become even more critical on large-diameter graphs such as \emph{Hyperlink-2012}, \emph{uk-2005} (web graph), and \emph{Road-USA} (road network), where larger number of stages makes per-stage load balance especially important and allows \ourgraph to achieve even greater performance gains.
\ourgraph underperforms the baselines in \textit{only} $2$ cases out of the $30$ experiments: \emph{Twitter PR} and \emph{Road-USA PR}.
Detailed explanations for these exceptions will be provided in \S\ref{subsec:ablation_study}.

We also compared \ourgraph against Pregel+~\cite{yan2015effective,lu2014large,yan2014pregel} and StarPlat~\cite{behera2024starplat}.
However, their implementations are at least $5\times$ slower than all the other baselines considered here and are therefore omitted from subsequent detailed analysis.

Fig.\ref{fig:perf_breakdown} illustrates the breakdown into computation, communication, and overhead for Twitter dataset on 16 machines.

\begin{figure}[t]
  \centering
  \includegraphics[width=0.6\columnwidth]{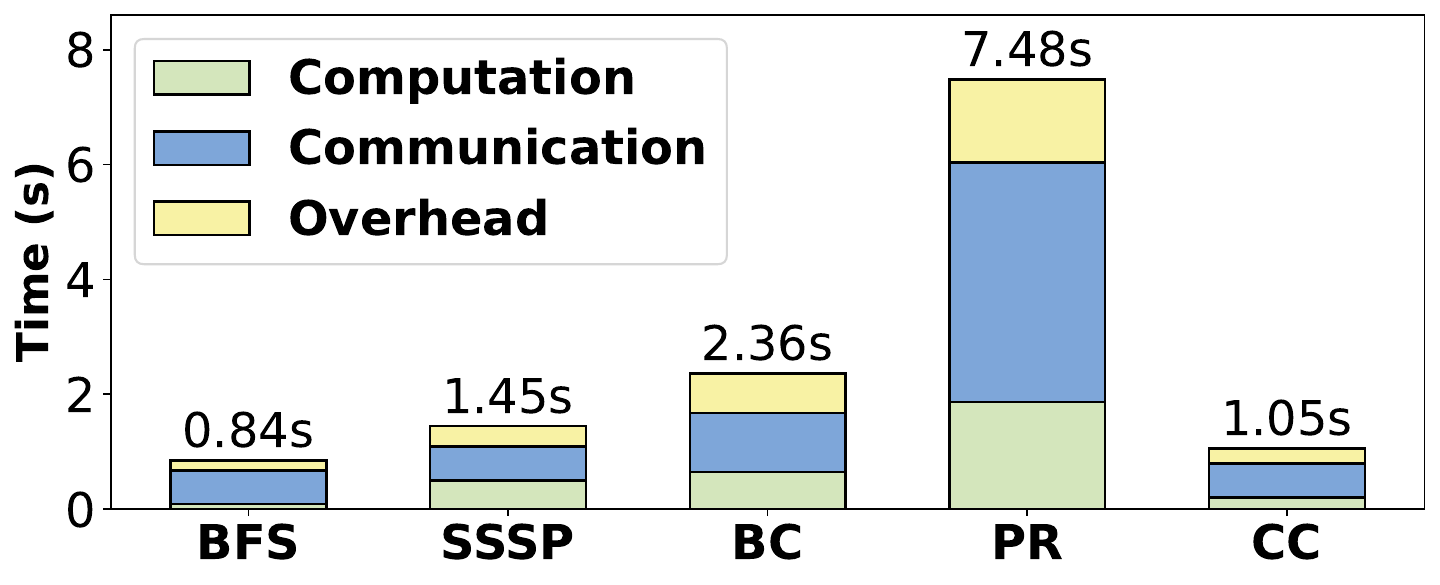}
  \vspace{-0.5em}
  \caption{Runtime breakdown on 16-machine Twitter.}
  \label{fig:perf_breakdown}
  \vspace{-1em}
\end{figure}

\subsection{Scalability}
\label{subsec:scalability}

\myparagraph{Weak Scaling}
We evaluate weak scaling on the graph systems by varying the number of machines while keeping the graph size per machine fixed.
To construct graphs on such a scale, we employ two graph generation methods: Erdős--Rényi (ER)~\cite{erd6s1960evolution}, which produces unskewed graphs, and Barabási--Albert (BA)~\cite{barabasi1999emergence}, which produces skewed graphs.
The average edge number per machine is fixed at 40M for fair comparison, where under 16-machine the graph is similar sizes as uk-2005.
For BA, we set the power-law exponent to $\gamma = 2.2$, consistent with the measured skew observed in natural graphs reported by PowerGraph~\cite{gonzalez2012powergraph}.
In Fig.\ref{fig:weakscale},
\ourgraph maintains nearly constant runtime across machine counts, showing excellent weak-scaling performance, whereas prior works incur significant scalability degradation.
See \suppref{sec:appendix_tdopg_weak_scaling} for results on other graph problems.

\myparagraph{Strong Scaling}
We also run strong scaling experiments using 1--16 machines.
We show in \suppref{sec:appendix_tdopg_strong_scaling} that \ourgraph exhibits good scalability across all configurations, while prior methods either show poor scalability when \#machines increases, or exhibit poor execution speed.

\begin{figure}[htbp]
  \centering
  \begin{subfigure}[t]{0.4\linewidth}
    \centering
    \includegraphics[width=\linewidth,trim=6pt 8pt 6pt 8pt,clip]{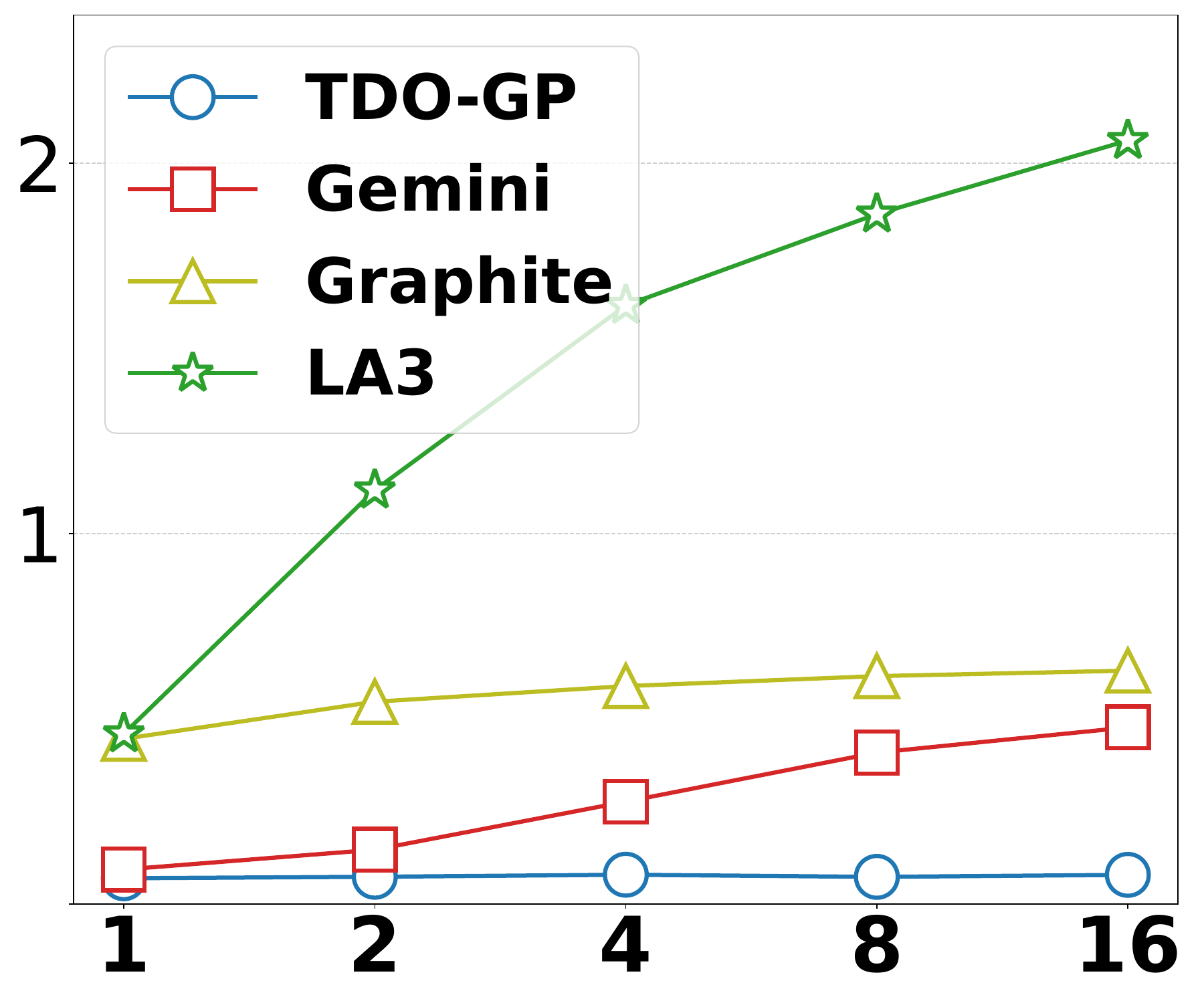}
    \caption{BC on ER Graph}
  \end{subfigure}
  \hspace{0.1em}
  \begin{subfigure}[t]{0.4\linewidth}
    \centering
    \includegraphics[width=\linewidth,trim=6pt 8pt 6pt 8pt,clip]{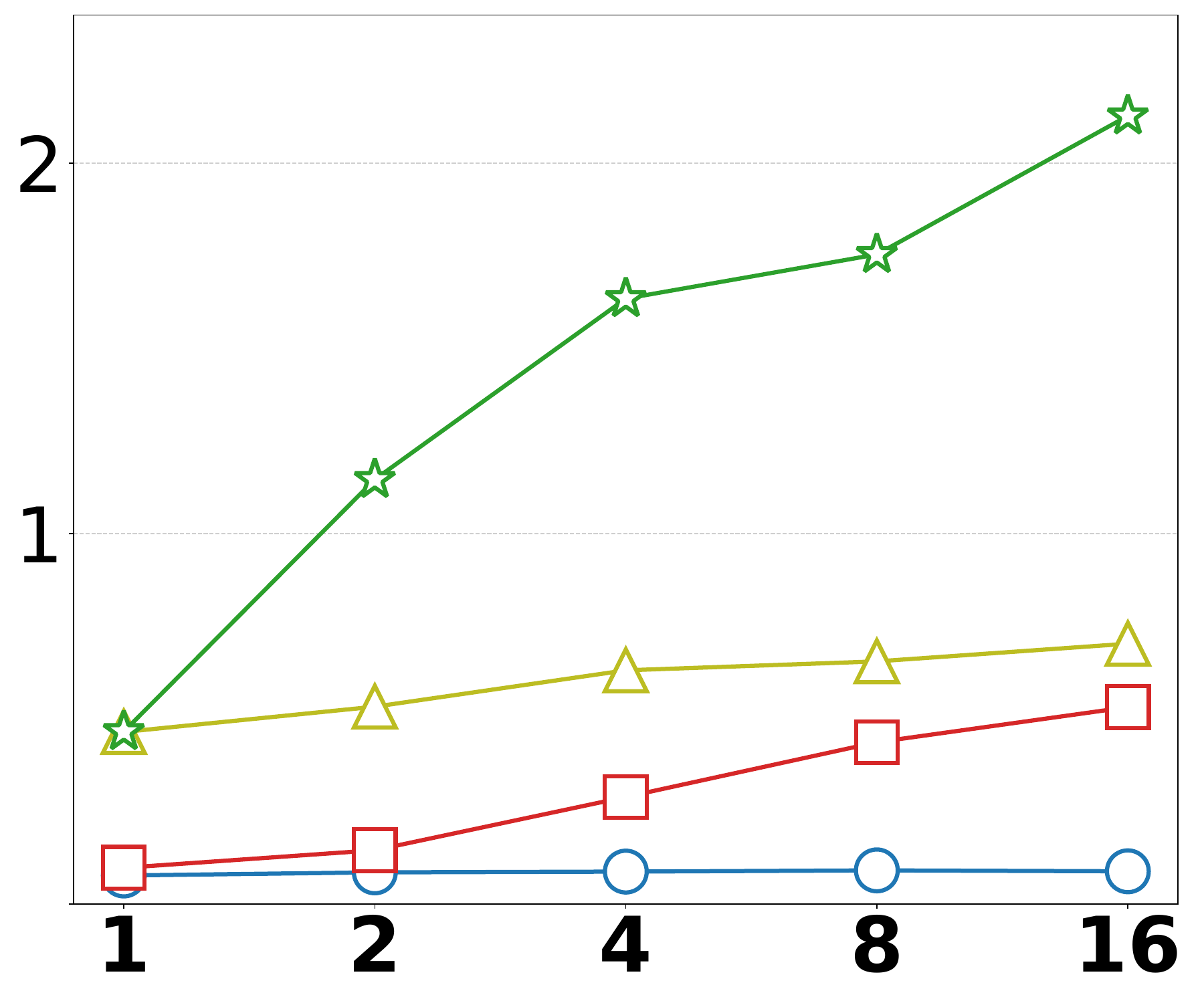}
    \caption{BC on BA Graph}
  \end{subfigure}
  \vspace{-0.2em}
  \caption{Weak scaling results. $y$-axis is runtime (in seconds) and $x$-axis is \#machines used. \ourgraph exhibits excellent weak-scaling performance and is significantly faster.}
  \vspace{-0.8em}
  \label{fig:weakscale}
\end{figure}

\subsection{Ablation Study}
\label{subsec:ablation_study}

\myparagraph{Sensitivity to \ourname}
We first show the effectiveness of \ourname within \ourgraph by comparing the final design with a Ligra plus direct-pull prototype (denoted as Ligra-Dist), which applies \S\ref{subsec:implementation_technique_outline} techniques but does not use \ourname.
Our BC results in Tab.\ref{tab:bc_twitter_comparison} on Twitter achieve up to $220\times$ speedup.

\ourgraph yields these gains through \textbf{\textit{improved load balance}}: in 16-machine Twitter-BC experiments, \ourgraph has a coefficient of variation of 0.061, compared with 0.321, 0.069, 0.331, 0.129 for Ligra-Dist, Gemini, Graphite, and LA3.

\begin{table}[h!]
\vspace{-0.5em}
\centering
\caption{Runtime (in seconds) of BC on Twitter shows that \ourname is a critical building block in \ourgraph.}
\label{tab:bc_twitter_comparison}
\vspace{-0.5em}
\footnotesize
\begin{tabular}{lcccc}
\toprule
\textbf{Number of Machines} & \textbf{1} & \textbf{4} & \textbf{8} & \textbf{16} \\
\midrule
Ligra-Dist (no \ourname) & 5.36 & 563.11 & 670.71 & 321.60 \\
\ourgraph   & \underline{\textbf{4.54}} & \underline{\textbf{3.70}} & \underline{\textbf{3.05}} & \underline{\textbf{2.35}} \\
\bottomrule
\end{tabular}
\vspace{-0.5em}
\end{table}

\myparagraph{Overhead of \ourname in \ourgraph}
Periodic \ourname execution is integrated into \ourgraph's graph initialization/update.
Our low-overhead implementation reduces this cost by up to $2.3\times$ compared with the metadata and execution-plan initialization/adjustment overheads of Gemini, Graphite and LA3.
See \suppref{sec:appendix_tdogp_init} for full details.

\myparagraph{Sensitivity to Graph-Specific Optimization}
We evaluate the sensitivity of \ourgraph to techniques in \S\ref{subsec:implementation_technique_outline}.
We measure the slowdown when removing each optimization family from the fully optimized \ourgraph.
Tab.\ref{tab:ablation_speedup} reports Twitter results with 16 machines, showing that disabling any single optimization family causes substantial performance degradation, with slowdowns of up to $4.09\times$, illustrating the importance of each technique in achieving high practical efficiency.
Similar trends appear for 2--15 machines and are omitted.
These graph-specific optimizations enable \ourgraph to outperform Gemini, which achieves moderate load balance.

\begin{table}[h!]
\centering
\footnotesize
\caption{Slowdown when removing graph-specific techniques (see \S\ref{subsec:implementation_technique_outline}) from the fully optimized \ourgraph.}
\label{tab:ablation_speedup}
\vspace{-0.5em}
\begin{tabular}{lccc}
\toprule
\textbf{Graph Problem} & \textbf{BFS} & \textbf{BC} & \textbf{PR} \\
\midrule 
Without (T1) & 1.89$\times$ & 1.79$\times$ & 1.67$\times$ \\
Without (T2) & 1.91$\times$ & 4.09$\times$ & 1.96$\times$ \\
Without (T3) & 1.26$\times$ & 1.33$\times$ & 2.50$\times$ \\
\bottomrule
\end{tabular}
\vspace{-0.5em}
\end{table}

\myparagraph{PageRank and NUMA Effects}
In Tab.\ref{tab:end_to_end_results}, the only 2 out of the 30 cases where \ourgraph underperforms prior systems are both PageRank (PR), and only against one prior baseline.
\ourgraph is built upon ParlayLib~\cite{blelloch2020parlaylib}, a fast parallel library with good scalability, but is NUMA-oblivious.
Each machine in our cluster contains four NUMA nodes connected in a square topology, where certain NUMA-to-NUMA accesses are significantly slower, hindering our PR performance.

In two ablations in \suppref{subsec:numa}, we show that (i) if the NUMA bottleneck is removed by using a machine with all-to-all NUMA connectivity, \ourgraph outperforms all baselines on PR by up to $2.3\times$ and (ii) even when restricted to a single machine, \ourgraph surpasses the optimized single-machine GBBS~\cite{dhulipala18scalable} by $1.5\times$.

\section{Conclusion}
We introduce the \textit{\taskdata orchestration} abstraction, unifying a variety of distributed applications.
We design \textit{\ourname} for this interface with low-overhead load balance.
Upon \ourname, we develop \textit{\ourgraph}, a distributed graph system. 
Our designs are theoretically grounded and practically efficient.

\bibliographystyle{ACM-Reference-Format}
\balance
\bibliography{ref}

\clearpage

\appendix

\renewcommand{\topfraction}{0.95}
\renewcommand{\bottomfraction}{0.95}
\renewcommand{\textfraction}{0.05}
\renewcommand{\floatpagefraction}{0.3}
\renewcommand{\dbltopfraction}{0.95}
\renewcommand{\dblfloatpagefraction}{0.4}
\setcounter{topnumber}{3}
\setcounter{bottomnumber}{2}
\setcounter{totalnumber}{4}
\setcounter{dbltopnumber}{2}

\section{Extended Related Work: Task- and Data-Centric Distributed Frameworks}
\label{sec:appendix_related_frameworks}

A large body of distributed programming frameworks expose task- and/or data-centric abstractions.
We survey the most prominent ones here and contrast them with \ourname.
The common thread is that these frameworks provide \emph{general} mechanisms for expressing distributed tasks and data, but either leave task--data co-location, communication aggregation, and skew-aware load balancing to the user, or do not target them at all.
\ourname instead targets this recurring pattern directly, providing a simpler task--data interface (\S\ref{subsec:abstraction_interface}) with automatic, load-balanced communication.

\myparagraph{General task-parallel and dataflow engines}
Systems such as Ray~\cite{moritz2018ray}, Dask~\cite{rocklin2015dask}, Spark~\cite{zaharia}, MapReduce~\cite{dean2008mapreduce}, Dryad~\cite{isard2007dryad} and DryadLINQ~\cite{yu2008dryadlinq}, CIEL~\cite{murray2011ciel}, Naiad~\cite{murray2013naiad}, Flink~\cite{carbone2015flink}, TensorFlow~\cite{abadi2016tensorflow}, parameter servers~\cite{li2014parameter}, Parsl~\cite{babuji2019parsl}, Swift/T~\cite{wozniak2013swift}, T-thinker~\cite{yan2019thinker}, and dependency-guided stream processing~\cite{kallas2022stream} provide general task scheduling over distributed data, and systems such as SCache~\cite{fu2018scache} accelerate the shuffle traffic they generate.
However, they do not directly optimize task--data co-location: users must still manage data partitioning, communication, and aggregation to avoid hotspots.

\myparagraph{Task-based HPC runtimes with explicit data placement}
Legion~\cite{bauer2012legion,bauer2021control,bauer2023visibility} and its language Regent~\cite{slaughter2015regent} and runtime Realm~\cite{treichler2014realm}, along with PaRSEC~\cite{bosilca2013parsec}, StarPU~\cite{augonnet2011starpu}, HPX~\cite{kaiser2014hpx}, Charm++~\cite{kale1993charm}, OmpSs~\cite{duran2011ompss}, Cilk~\cite{Cilk95,schardl2023opencilk}, Kokkos~\cite{edwards2014kokkos}, RAJA~\cite{beckingsale2019raja}, Pure~\cite{psota2024pure}, and Sequoia~\cite{fatahalian2006sequoia} provide richer data-placement and task-dependency control.
But they typically require users to express lower-level region/partition mappings or explicit placement and dependency annotations, whereas \ourname requires only a batch of lambda-task contexts that name their target data.
Individual mechanisms in this line target our objectives in isolation, such as distributed work stealing for dynamic load balance~\cite{parikh2016matchmaking} and load-balance-aware data placement~\cite{xie2023merchandiser}, but leave the task--data interface to the programmer.

\myparagraph{Partitioned global address space (PGAS) languages}
PGAS languages and libraries including Chapel~\cite{chamberlain2007chapel}, X10~\cite{charles2005x10}, UPC++~\cite{zheng2014upcxx}, Titanium~\cite{yelick1998titanium}, Co-array Fortran~\cite{numrich1998coarray}, Global Arrays~\cite{nieplocha2006global}, and OpenSHMEM~\cite{chapman2010openshmem} give programmers explicit control over data distribution and locality.
This again shifts partitioning and communication-aggregation decisions to the user, rather than providing automatic load-balanced communication as \ourname does.

\myparagraph{Data-centric dataflow optimization}
Data-centric frameworks such as DaCe and Data-Centric Python~\cite{bennun2019dace,ziogas2021productivity}, together with Legate NumPy~\cite{bauer2019legate}, Halide and its distributed extension~\cite{ragankelley2013halide,denniston2016halide}, and TVM~\cite{chen2018tvm}, optimize data movement via dataflow transformations.
However, they do not target runtime skew-aware load balancing nor hotspot communication reduction as \ourname does.

\myparagraph{Summary}
The frameworks most closely related to \ourname are the general task/data-centric systems above---in particular Ray, Dask, Legion, and Data-Centric Python.
Relative to all of them, \ourname \romenum{1} targets task--data co-location and skew-aware, hotspot-avoiding load balance directly rather than delegating it to the user, and \romenum{2} exposes a substantially simpler interface: the user supplies only lambda-tasks and the addresses of their target data (\S\ref{subsec:abstraction_interface}), and the runtime automatically handles partitioning, load-balanced communication, and aggregation.

\section{Extended Related Work: Models for Distributed Computing}
\label{sec:model}

In this paper, we model and analyze algorithms using the classic 
Bulk-Synchronous Parallel (BSP) model~\cite{valiant1990bridging,valiant1990general,valiant1988optimally}.
The BSP consists of $P$ machines, each with a sequential processing element (PE) and a local memory.
There is no shared global memory.
All machines are connected by a point-to-point communication network.
All remote accesses to data on other machines are carried out via messages sent over this network.
The system operates in bulk-synchronous rounds (\textit{supersteps}), where periodic barrier synchronizations separate periods of asynchronous computation and communication.

\myparagraph{Cost Metrics}
The BSP model defines the notion of an $h$-relation of a superstep, which combines the maximum computation work, maximum communication amount and synchronization time over all machines.
Due to using \textit{maximum} when defining $h$-relations, load balance is critical for achieving good performance.
In all our analysis, for clarity, we decouple the \textbf{\textit{computation time}} and \textbf{\textit{communication time}}, and do not consider the BSP parameters $g$ and $L$ in~\cite{valiant1990bridging,gerbessiotis1994direct}.
Computation time $t$ and communication time $h$ can be combined into a formal BSP result $\Theta(gh+t)$ as long as the synchronization overhead can be hidden by system communication, i.e., $(gh+t)\in\Omega(L)$.

\myparagraph{Other Distributed Models}
The Massively Parallel Model of Computation (MPC)~\cite{im2023massively,karloff2010model,goodrich2011sorting,beame2013communication,andoni2014parallel} family of models regards synchronization as the main bottleneck in distributed systems and takes \textit{communication rounds} as the only optimization target.
We choose the the BSP model because it differs from MPC in two ways:
\romenum{1} The the BSP model accounts for the total communication among machines rather than ignoring the amount of communication in each round. 
\romenum{2} The the BSP model also considers computation costs neglected by the MPC model, which turn out to be critical for real-world system designs on emerging distributed hardware~\cite{kang2022pimtreevldb,bindschaedler2020hailstorm} other than MapReduce systems ~\cite{dean2008mapreduce}.

LOCAL/CONGEST models~\cite{linial1987distributive,linial1992locality} are a popular model family for distributed systems.
Like the models we consider, each machine has only a local view of the computation.
These models consider network topology, which we ignore in the BSP model. 
However, they have similar weaknesses to the MPC model with regards to their metrics: (i) each machine can send/receive up to its maximum allowed communication at each round, whereas we seek to minimize communication, and (ii) computation costs are not considered.

The Processing-in-Memory (PIM) model~\cite{kang2021processing} is a recently proposed model for emerging compute-in-memory hardware~\cite{mutlu2023primer}.
The model has a powerful Host CPU and $P$ PIM nodes, each comprised of a PE and a memory.
It operates in bulk-synchronous rounds, but unlike the BSP model, all communication must be sent by or to the Host CPU.
Furthermore, the model decouples the computation time of the CPU from other processors, resulting in a lack of analysis of load-balance between them.

\section{Proof for \Cref{theorem:main}}
\label{appendix:main_proof}
We will use the balls-into-bins game of \Cref{lemma:weighted_balls_into_bins} to prove load balance.

\begin{lemma}[Weighted Balls into Bins ~\cite{sanders1996competitive,nievergelt1973binary}]
\label{lemma:weighted_balls_into_bins}
Placing weighted balls with total weight $W=\Sigma w_i$ and weight limit $W/(P\log P)$ into $P$ bins uniformly randomly yields $\Theta(W/P)$ weight in each bin \whp.
\end{lemma}

\begin{proof}[Proof for \Cref{theorem:main}]
    When $n=\Omega(P^\epsilon)$ for some $\epsilon>1$, things become trivial by setting $F=\Theta(n/P)$.
    Here we only analyze the hard case when $n=O(P\log^\epsilon{P})$ for some $\epsilon>1$.
    We set $C=\Theta(B/\sigma)$ and $F=\Theta(\log{P}/\log{\log{P}})$.
    \romenum{1} Phase 1 takes at most $O(\log_{F}{P})$ rounds and in each round at most $O(n)$ words of messages are sent.
    Since meta-task sizes are bounded by $O(C\log_{C}{n})\leq{O}\left(\frac{n}{P\log{P}}\right)$, the weight limit in \Cref{lemma:weighted_balls_into_bins} is met and the communication is thus load-balanced.
    When $n=\Omega\left(\frac{PB\log^{3}{P}}{\log\log{P}}\right)$, the communication time is $O(\frac{n}{P}\log_F{P})$ $=O(\frac{n}{P}\log_{\frac{n}{P}}{P})$ in phase 1.
    Phase 2 takes less communication time than phase 1 because it only requires partial communication over phase 1's trace, so \romenum{1} is proved.
    \romenum{2} Since meta-tasks have size $O(C\log_{C}{n})$, \Cref{lemma:weighted_balls_into_bins} can also be applied.
\end{proof}

Moreover, each machine holds $\Theta(n/P)$ tasks when a stage ends. 
Thus, for iterative applications, where each executed task in a stage dictates a task to be executed in the next stage, the initial condition of \Cref{theorem:main} that each machine holds $\Theta(n/P)$ tasks when a stage begins, is inductively met.
This implies that \Cref{theorem:main} can again be applied to the next stage.
Note that such iterative applications include all the supported applications listed in \S\ref{subsec:abstraction_application}.

\section{Additional Ablations for \ourname in the \KV Store}
\label{sec:appendix_kv_ablation}
Complementing the \KV store evaluation in \S\ref{sec:key_value_case_study}, we present three additional ablations of \ourname, all using the same \KV/YCSB setup as Fig.\ref{fig:gamma_linecharts}: a runtime component breakdown, a memory-overhead study, and a tree-parameter sensitivity study.

\subsection{Component Time Breakdown}
\label{sec:appendix_component_breakdown}
We evaluate \ourname on \KV store (\S\ref{sec:key_value_case_study}) and decompose its overhead into three components:
\romenum{1} \textit{task push}: task sorting, meta-task construction and merging, and tree-structured aggregation communication;
\romenum{2} \textit{data pull}; and
\romenum{3} \textit{writeback propagation}.
Tab.\ref{tab:component_breakdown} reports a representative breakdown at $16$ machines with $2$M tasks per machine (averaged over workloads A/C/LOAD and $\gamma{=}1.5$/$2.5$).
The task-push phase incurs certain overhead, whereas the orchestrated execution scheme makes the subsequent data pull and writeback highly efficient.
Although orchestration introduces overhead, this cost enables significant performance improvements and avoids the contention-induced performance degradation observed in other baselines, which leads to the respective averaged speed-ups of $1.42\times$--$2.83\times$ as shown in \S\ref{sec:key_value_case_study}.

\begin{table}[h!]
\centering
\footnotesize
\caption{\ourname tree-mode component breakdown on \KV ($16$ machines, $2$M tasks/machine).}
\label{tab:component_breakdown}
\vspace{-0.5em}
\begin{tabular}{lc}
\toprule
\textbf{Component} & \textbf{Share of runtime} \\
\midrule
Task push & $60.1\%$ \\
Data pull and Execution & $34.8\%$ \\
Writeback propagation & $5.1\%$ \\

\bottomrule
\end{tabular}
\vspace{-0.5em}
\end{table}

\subsection{Memory Overhead}
\label{sec:appendix_memory_overhead}
We profile the peak memory usage of \ourname against the Direct Push and Direct Pull baselines on the \KV store using $16$ machines with $2$M tasks per machine.
\ourname maintains additional meta-task metadata to enable contention migration, incurring a modest peak-memory overhead over the most memory-efficient baseline, Direct Pull: $9.7\%$ on average, increasing to $15.5\%$ under the highest skew ($\gamma{=}1.5$) and decreasing to $4.0\%$ at $\gamma{=}2.5$.
Direct Push remains within $\pm2\%$ of Direct Pull.
In absolute terms, the tree mode peaks at $2.06$~GB per machine, compared with $1.8$--$2.0$~GB for Direct Pull.
Thus, the metadata used by \ourname to achieve load balance incurs only a modest $4.0\%$--$15.5\%$ memory overhead, which is well justified by its significant runtime gains.

\subsection{Hyperparameter Sensitivity}
\label{sec:appendix_param_sensitivity}
\ourname uses two hyperparameters---the communication-forest fanout $F$ (\S\ref{subsubsec:comm_forest}) and the meta-task-set capacity $C$ (\S\ref{subsubsec:meta_task}).
We sweep $F\in\{2,4,8\}$ and $C\in\{2,3,4,8,16\}$ on the \KV store using $16$ machines, while keeping the remaining Fig.\ref{fig:gamma_linecharts} setup fixed.
Tab.\ref{tab:param_sensitivity} reports the geometric-mean runtime of each stable configuration, normalized to the $F{=}2,C{=}3$ reference.
When $F$ and $C$ remain within the same asymptotic regime defined in \S\ref{sec:main_design}, runtime stays within a narrow $0.95\times$--$1.17\times$ range, varying by less than $20\%$ across all stable configurations.
This stability indicates that parameter choices satisfying the same asymptotic conditions yield consistently strong practical performance, highlighting the value of the theoretical guidance established in \S\ref{subsec:orch_formal_guarantee}.
The best setting is $F{=}4,C{=}16$ ($0.951\times$; absolute mean $0.353$s).
Only extreme configurations with a very small capacity relative to the fanout regress or become unstable under the highest skew, which our default parameterization avoids.

\begin{table}[h!]
\centering
\footnotesize
\caption{Sensitivity of \ourname to the communication-forest fanout $F$ and meta-task-set capacity $C$ on \KV store ($16$ machines). Each cell is the geometric-mean runtime normalized to the $F{=}2,C{=}3$ reference (lower is better; best in \textbf{bold}). Across all measured settings runtime stays within $0.95\times$--$1.17\times$ (less than $20\%$ spread); ``--'' marks settings outside the reported stable sweep.}
\label{tab:param_sensitivity}
\vspace{-0.5em}
\begin{tabular}{c ccccc}
\toprule
& \multicolumn{5}{c}{\textbf{Meta-task-set capacity $C$}} \\
\cmidrule(lr){2-6}
\textbf{Fanout $F$} & $2$ & $3$ & $4$ & $8$ & $16$ \\
\midrule
$2$ & -- & $1.000$ & -- & $1.169$ & $0.982$ \\
$4$ & $1.006$ & -- & -- & $0.954$ & $\mathbf{0.951}$ \\
$8$ & $0.966$ & $0.958$ & $1.011$ & $1.018$ & $0.954$ \\
\bottomrule
\end{tabular}
\vspace{-0.5em}
\end{table}

\section{Details for Graph-Specific Implementation Techniques}
\label{appendix:implementation_techniques}

We present here the implementation techniques used in our graph processing system, which efficiently exploit the data layout and execution flow provided by \ourname, leading to the strong performance demonstrated in \S\ref{sec:eval}.

\subsection{Optimized Global Communication}
\label{appendix_subsubsec:global_comm}

As we will show in \S\ref{subsec:ablation_study}, the critical bottleneck in a distributed graph system is the inter-machine communication time.
Optimizing global communication is therefore essential.
First, we apply deduplication during message passing along the meta-task trees, aggressively merging all repeated items.
Second, in dense mode, we reduce broadcast traffic by sending the values of each vertex only to the machines that store its corresponding edges.
The mapping of vertices to machines is established during the preprocessing phase.

\subsection{Work-Efficient Local Computation}
\label{appendix_subsubsec:local_comp}

For all local computations in our implementation, we employ work-efficient designs that are both theoretically optimal and practically fast.
Our design leverages ParlayLib~\cite{blelloch2020parlaylib}, a state-of-the-art C++ parallel library providing a wide range of work-efficient parallel algorithms as primitives.

The \DISTVERTEXSUBSET structure consists of a collection of \VERTEXSUBSET instances, one per machine, where each \VERTEXSUBSET can independently switch between two representation formats introduced in Ligra~\cite{ShunB2013}.
We introduce two implementation improvements.
\romenum{1} In the sparse representation, we replace their array of vertices with a phase-concurrent hash table with adaptive sizing, enabling more efficient lookups within \VERTEXSUBSET.
\romenum{2} In the dense representation, we replace the original parallel C++ Boolean-map with a concurrent bitmap, improving cache efficiency.

For local edges whose source and destination vertices reside on the same machine, we bypass the global orchestration framework and instead execute a local \EDGEMAP.
We implement a lightweight \EDGEMAP entirely using ParlayLib primitives, which, as we show in \S\ref{subsec:numa}, is more efficient than the ones in Ligra/GBBS for \EDGEMAP workloads.

Additionally, we allow users to provide an optional input function \textit{filter\_dst} to \DISTEDGEMAP to accelerate execution by filtering out vertices that are guaranteed not to be destination vertices in the current round.
Disabling \textit{filter\_dst} does not affect correctness but may reduce efficiency.

\subsection{Aligned Coordination between Components}
\label{appendix_subsubsec:aligned_coord}

A distributed graph processing system involves multiple software components running across diverse hardware resources.
Ensuring a coordinated execution schema among these components is crucial.
Our implementation optimizes this coordination to achieve aligned and efficient execution.

For better alignment with the parallel thread scheduler in ParlayLib, we avoid using nested \textsc{ParallelFor}, which is commonly seen in graph algorithms but inefficient under such schedulers.
Instead, we adopt \textsc{Flattern}, assisted by \textsc{delayed} structures~\cite{blelloch2020parlaylib}, to transform two-dimensional nested parallelism into one-dimensional parallelism.

To better align local execution across threads and minimize cache thrashing, we design concurrent counters and other shared concurrent structures with sizes padded to a multiple of the cache-line size.
This padding prevents false sharing, ensuring that accesses from different cores do not contend on the same cache line.

We also align local computation with the load-balanced data layout across machines.
We initialize the graph with a vertex-value layout such that the total number of outgoing edges (out-degrees) assigned to each machine is approximately equal.
This ensures that most local computations are naturally load-balanced across machines, since each handles a comparable number of edges.

\section{Additional Scalability Evaluation for \ourgraph in Graph Processing}
\label{sec:appendix_graph_scalability}

\subsection{Weak-Scaling Scalability of \ourgraph}
\label{sec:appendix_tdopg_weak_scaling}

In addition to the BS results presented in \S\ref{subsec:scalability}, we show the weak-scaling results for PR in Fig.\ref{fig:weakscale_appendix}.
Similarly, \ourgraph again maintains a nearly constant runtime as the machine count increases, demonstrating excellent weak scalability, while prior works suffer substantial performance degradation.

\begin{figure}[htbp]
  \centering
  \begin{subfigure}[t]{0.4\linewidth}
    \centering
    \includegraphics[width=\linewidth,trim=6pt 8pt 6pt 8pt,clip]{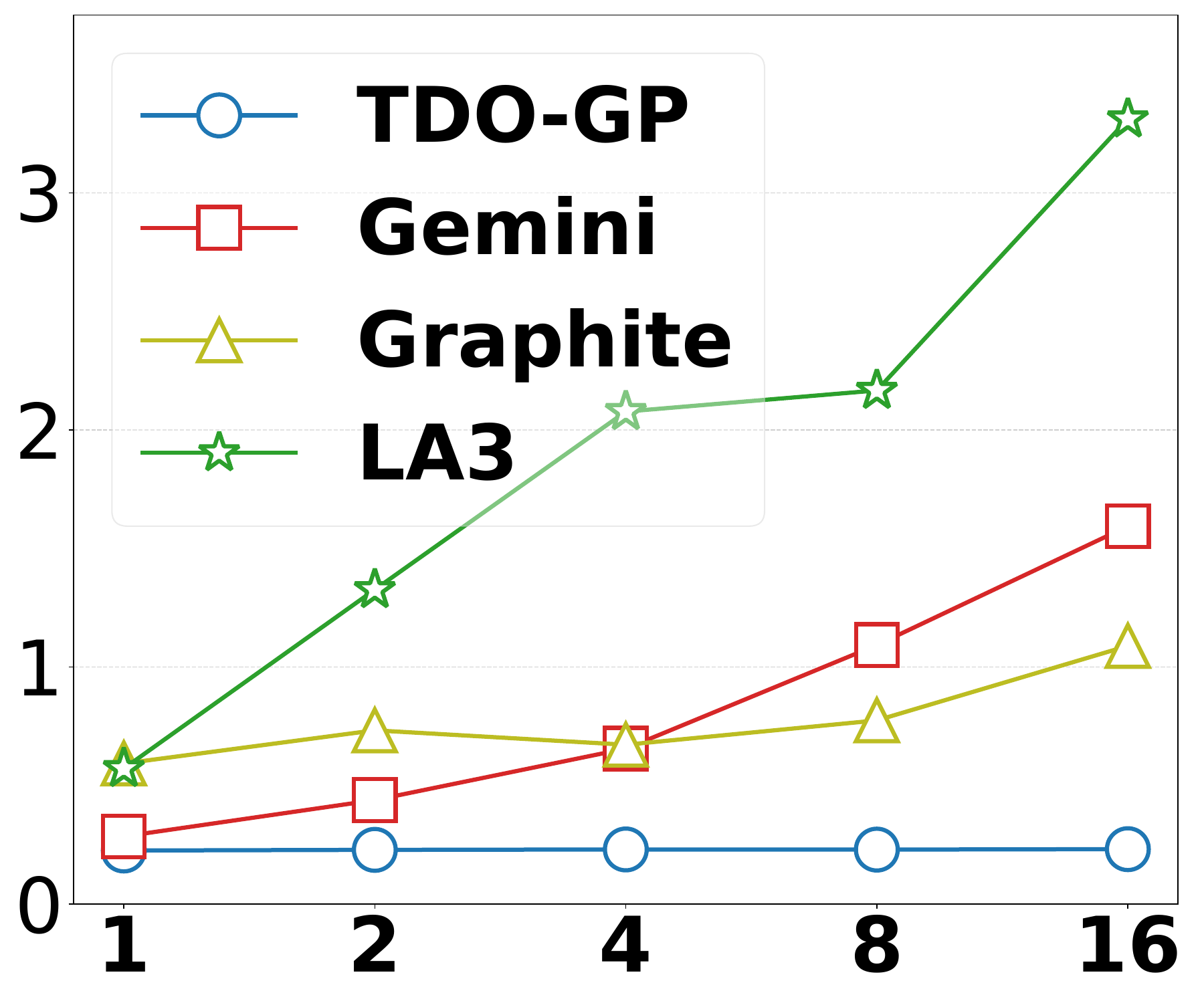}
    \caption{PR on ER Graph}
  \end{subfigure}
  \hspace{0.1em}
  \begin{subfigure}[t]{0.4\linewidth}
    \centering
    \includegraphics[width=\linewidth,trim=6pt 8pt 6pt 8pt,clip]{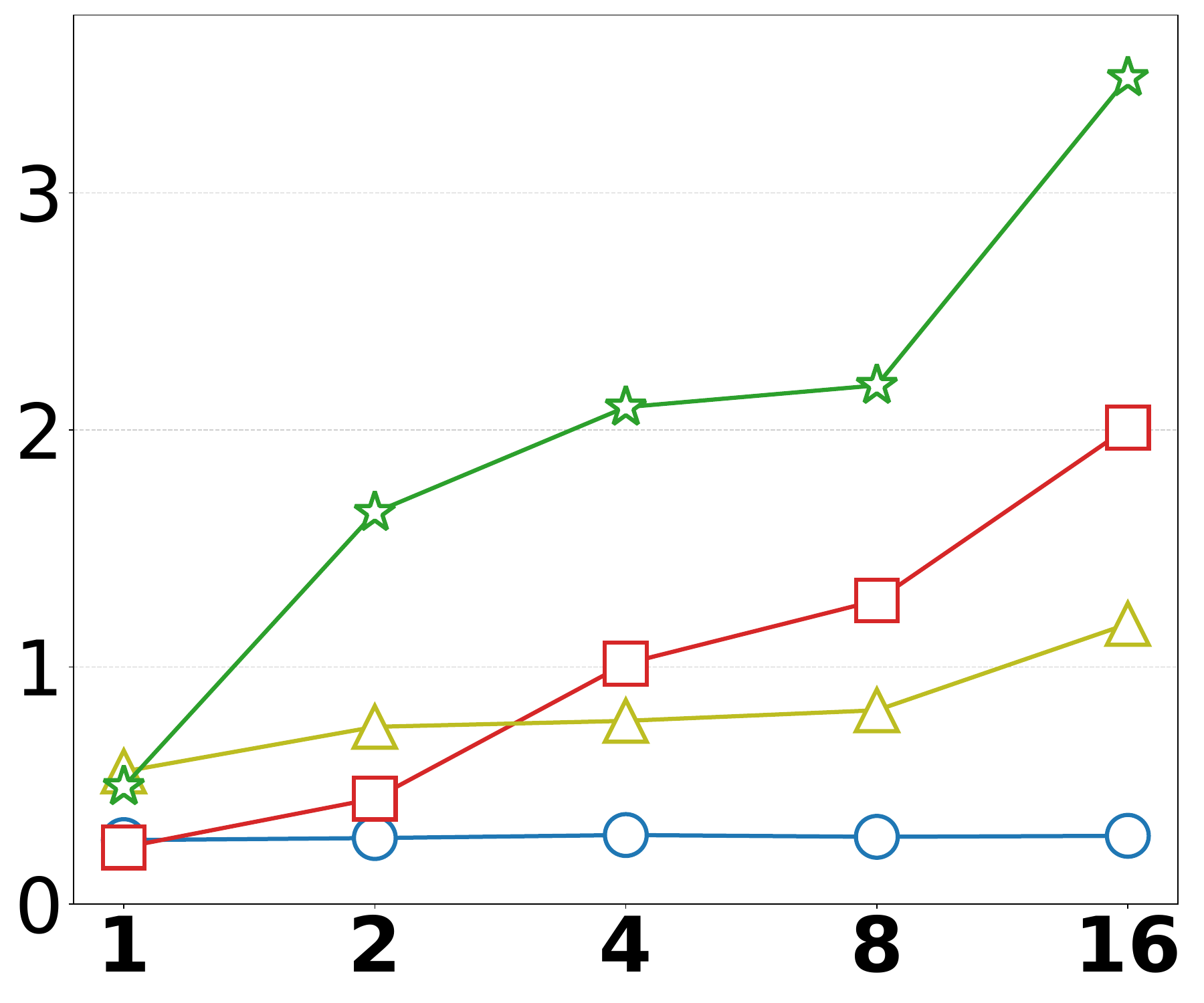}
    \caption{PR on BA Graph}
  \end{subfigure}
  \caption{Weak scaling results. $y$-axis is runtime (in seconds) and $x$-axis is \#machines used. \ourgraph exhibits excellent weak-scaling performance and is significantly faster.}
  \vspace{-0.8em}
  \label{fig:weakscale_appendix}
\end{figure}

\subsection{Strong-Scaling Scalability of \ourgraph}
\label{sec:appendix_tdopg_strong_scaling}
We run the SSSP and BC algorithms on the Twitter dataset using 1--16 machines.
Fig.\ref{fig:twitter_strong_scale_right} shows that \ourgraph exhibits good scalability across all configurations, while other methods either show poor scalability when increasing the number of machines, or exhibit poor performance.
\begin{figure}[ht]
  \centering
  \subcaptionbox{SSSP on Twitter.\label{fig:sssp-twitter}}[
    0.48\linewidth]{\includegraphics[width=\linewidth]{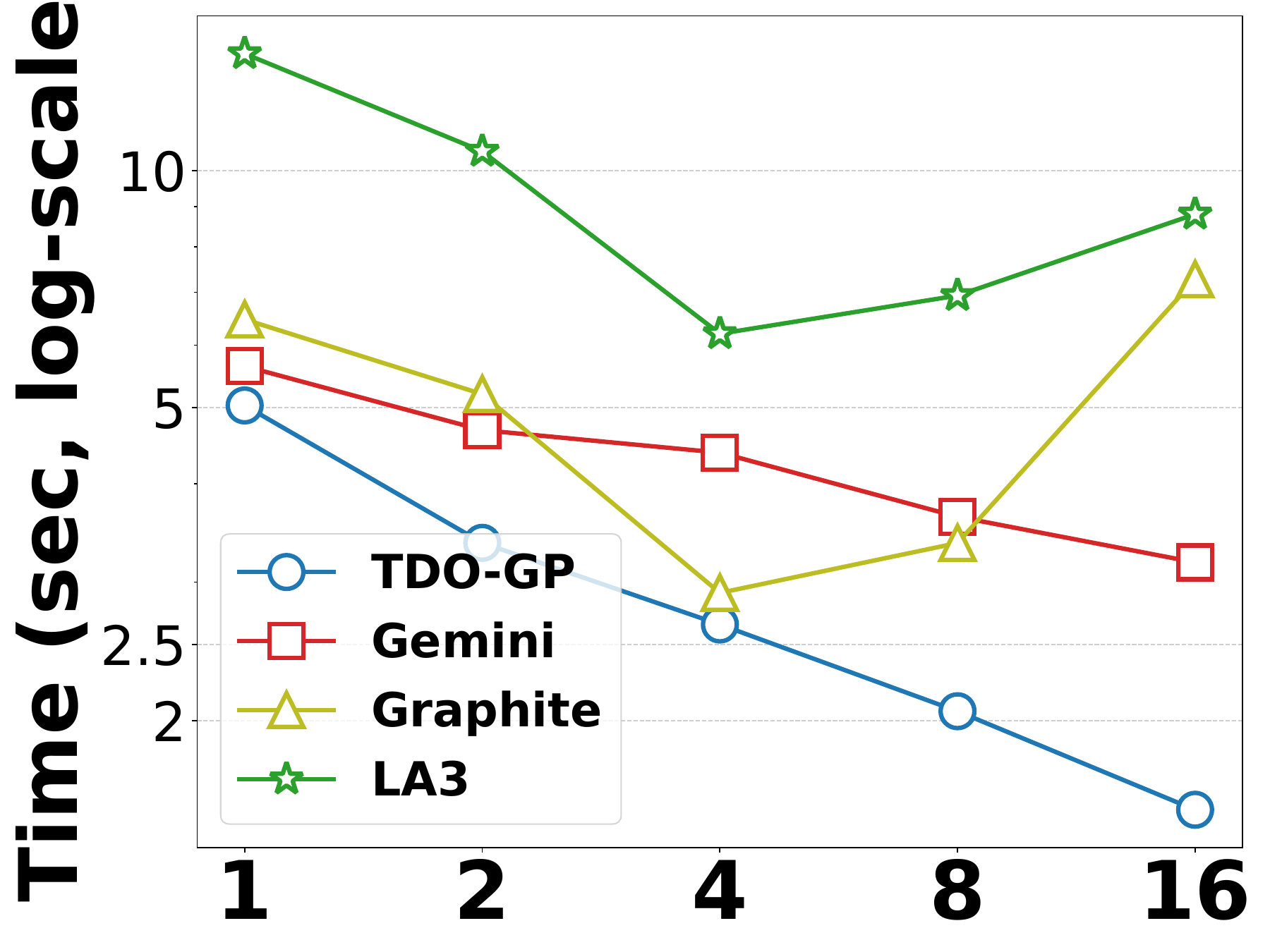}}
  \hfill
  \subcaptionbox{BC on Twitter.\label{fig:bc-twitter}}[
    0.48\linewidth]{\includegraphics[width=\linewidth]{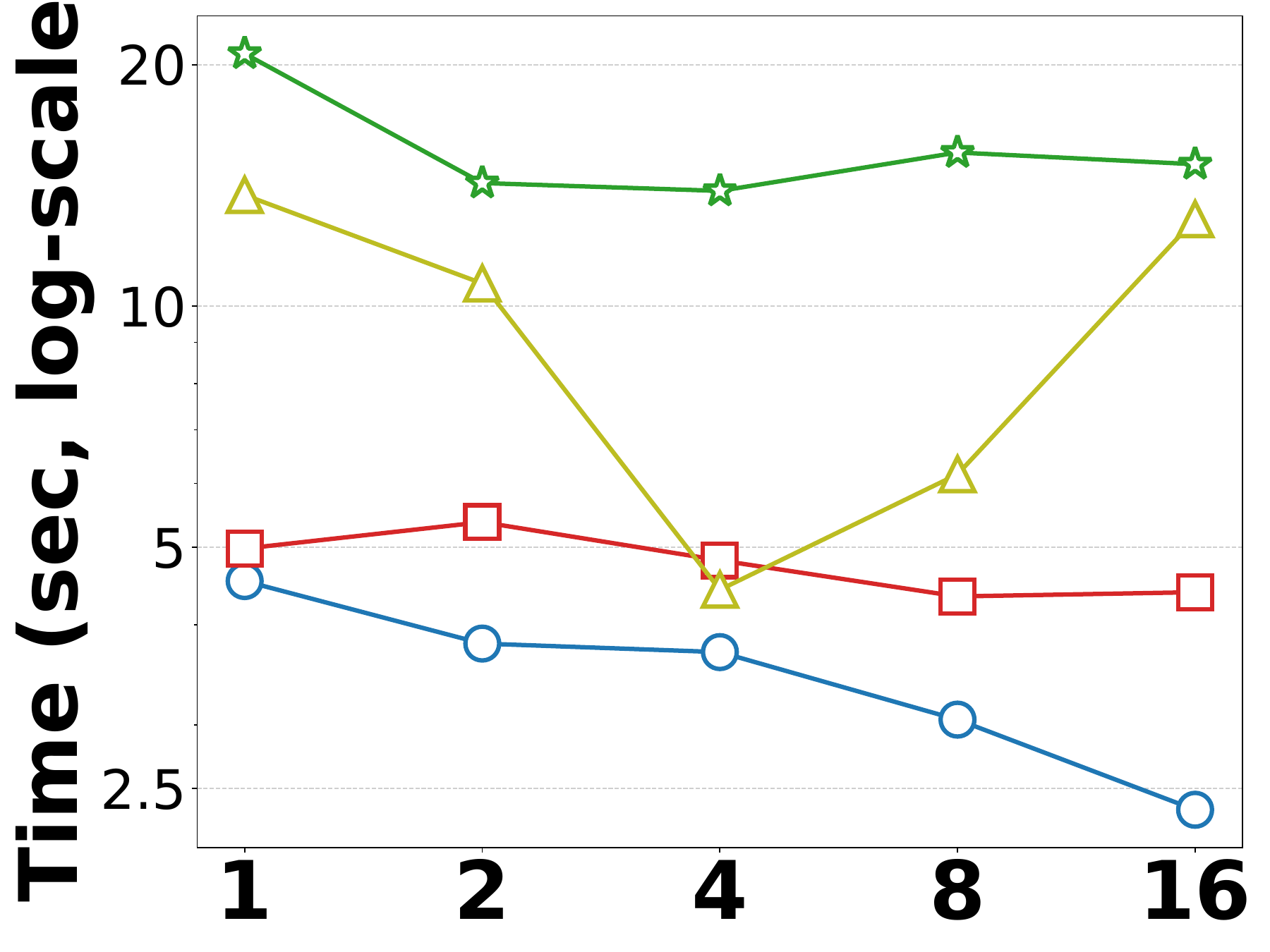}}
        \vspace{-0.5em}
        \caption{Strong scaling results on Twitter dataset with 1--16 machines. \ourgraph is significantly faster and more scalable.}
        \vspace{-0.8em}
        \label{fig:twitter_strong_scale_right}
\end{figure}

\section{Additional Ablations for \ourgraph in Graph Processing}
\label{sec:appendix_graph_ablation}

\subsection{Initialization and Updating Overhead}
\label{sec:appendix_tdogp_init}
In \ourgraph, periodic \ourname re-orchestration is integrated into graph initialization and updating (\S\ref{subsec:overall_design_outline}): the two-stage periodic process constructs the owner-partition data structures (i.e., the source and destination trees) once and reuses these pre-orchestrated flows across many subsequent \DISTEDGEMAP stages.
The overhead of this design therefore depends on the efficiency of the underlying initialization.
We measure the initialization/update cost of \ourgraph's optimized owner-partition CSR construction at $16$ machines and compare it with the data-structure loading/updating times of prior baselines.
As Tab.\ref{tab:tdogp_init} shows, \ourgraph incurs lightweight initialization and updating overhead: on the largest and most important datasets, it is faster than the loading times of Gemini, Graphite, and LA3---e.g., $50.7$s versus $69.2$/$70.1$/$157.3$s on Twitter and $58.1$s versus $131.6$/$99.7$/$164.2$s on Friendster.

\begin{table}[h!]
\centering
\footnotesize
\caption{\ourgraph initialization/updating time (in seconds) at $16$ machines vs.\ baselines (fastest is underlined).}
\label{tab:tdogp_init}
\vspace{-0.5em}
\begin{tabular}{lrrrr}
\toprule
\textbf{Dataset} & \textbf{\ourgraph} & \textbf{Gemini} & \textbf{Graphite} & \textbf{LA3} \\
\midrule
Twitter    & \underline{\textbf{50.68}} & 69.16  & 70.14  & 157.29 \\
Friendster & \underline{\textbf{58.08}} & 131.61 & 99.70  & 164.23 \\
Road-USA   & \underline{\textbf{0.90}}  & 2.08   & 4.60   & 3.43   \\
\bottomrule
\end{tabular}
\vspace{-0.5em}
\end{table}

\subsection{Ablation on PageRank and NUMA Effects}
\label{subsec:numa}
We conduct two additional ablation studies to evaluate the effectiveness of \ourgraph when NUMA-related limitations are mitigated.
First, we perform a multi-machine experiment on the same cluster in \S\ref{subsec:e2e_results}, but restrict execution to a single NUMA node per machine.
The results (Tab.\ref{tab:twitter_pr_numa_comparison}) show that the PR performance gap between \ourgraph and baselines is substantially reduced.
Moreover, \ourgraph has a significantly smaller memory footprint than Graphite.
The remaining minor gap is because ParlayLib is optimized for many cores, making the single-NUMA configuration less favorable.

\begin{table}[h!]
\centering
\footnotesize
\caption{Runtime (in seconds) of PR on Twitter. Each machine uses only $1$ NUMA node. OOM is out-of-memory.}
\label{tab:twitter_pr_numa_comparison}
\vspace{-0.5em}
\begin{tabular}{lcccc}
\toprule
\textbf{Machine Num.} & \textbf{1} & \textbf{4} & \textbf{8} & \textbf{16} \\
\midrule 
Gemini   & 46.67 & 30.72 & 25.61 & 18.59 \\
Graphite & OOM   & OOM   & 22.18 & 14.69 \\
\ourgraph & 40.74 & 32.30 & 21.67 & 15.07 \\
\bottomrule
\end{tabular}
\vspace{-0.7em}
\end{table}

Second, we conduct a single-machine experiment on a more powerful server with four Intel(R) Xeon(R) E7-8867 v4 @2.40GHz processors, providing a total of 144 cores and 180MB of L3 cache.
Unlike our 1024-core cluster, it features an all-to-all interconnection across its four NUMA nodes, eliminating the bottlenecks of square-topology interconnects.
Tab.\ref{tab:twitter_results_on_aware} shows that \ourgraph consistently outperforms all other baselines by a significant margin on Twitter.
Notably, \ourgraph restricted to a single machine even surpasses GBBS~\cite{dhulipala18scalable}, the state-of-the-art single-machine implementation (suitable \textit{only} for a single machine), owing to our optimized lightweight local \EDGEMAP.

\begin{table}[H]
\centering
\footnotesize
\caption{Twitter runtime (in sec) on one Xeon E7-8867v4.}
\label{tab:twitter_results_on_aware}
\vspace{-0.5em}
\begin{tabularx}{0.4\textwidth}{lYYY}
\toprule
\textbf{Method} & \textbf{BFS} & \textbf{BC} & \textbf{PR} \\
\midrule
Gemini   & 0.29 & 1.17 & 4.72 \\
Graphite & 2.17  & 3.91 & 10.05 \\
GBBS     & 0.29 & 1.05  & 6.81  \\
\ourgraph  & \underline{\textbf{0.12}} & \underline{\textbf{0.95}} & \underline{\textbf{4.43}}  \\
\bottomrule
\end{tabularx}
\vspace{-0.2em}
\end{table}

Overall, \ourgraph outperforms prior baselines on 28 of 30 cases on our original cluster and these ablation studies show that \ourgraph is competitive or better even on the remaining 2 cases on clusters with fast NUMA interconnects.

\section{Algorithms for Graph Problems}
\label{sec:appendix_bc_alg}
Alg.\ref{alg:bfs_edgemap} presents pseudocode for breadth-first search (BFS), offering a more efficient implementation via \DISTEDGEMAP than directly using the orchestration interface.

In addition, Alg.\ref{alg:bc} is presented here to describe the pseudocodes for the betweenness centrality (BC) computation from a single root vertex, which is commonly used in performance tests in prior works~\cite{ShunB2013,dhulipala18scalable,zhu2016gemini}.

We also provide the actual C++ implementation of the BC algorithm in Fig.\ref{fig:bc_cplus_code} using the \DISTEDGEMAP interface.
Despite the complexity of the sophisticated, work-efficient BC algorithm from GBBS~\cite{dhulipala18scalable} we follow, our code remains highly concise, requiring fewer than $70$ lines.

\begin{algorithm}[t]
\small
\fontsize{9pt}{12pt}\selectfont
\SetKwFor{IF}{If}{then}{endif}
\SetKwFor{While}{While}{do}{endwhile}
\SetKwFor{ParFor}{parallel for}{do}{endfch}
\SetKwFor{For}{for}{do}{endfch}
\SetKwFor{Function}{Function}{:}{}
\KwIn{$G=(V,E)$, where $|V|=n$; \textsc{StartVertex}}
\KwOut{\textsc{Dist}[$1:n$]: Distrib.~vertex value array}
Initialize Distrib.~vertex value array: \textsc{Dist}[$1:n$] = [-1]$\times$n\\
\textsc{round}$=1$\\
\Function{f ( (u,v): Edge )}{
    \Return \textsc{round} \textsc{IF} (\textsc{Dist}[$u$] == \textsc{round} - 1) \textsc{ELSE} -1
}
\Function{write\_back ( v: Vertex, aggr\_value )}{
    \IF{\textsc{Dist}[$v$] == -1 and aggr\_value != -1}{
        \textsc{Dist}[$v$] = aggr\_value \\
        \Return True
    }
    \Return False
}
\Function{BFS ( $G$, \textsc{StartVertex} )}{
    \DISTVERTEXSUBSET \textsc{frontier} = \{\textsc{StartVertex}\} \\
    \textsc{Dist}[\textsc{StartVertex}] = $0$ \\
    \While{\textsc{frontier} is not empty}{
        \textsc{frontier} = \DISTEDGEMAP($G$, \textsc{frontier}, f, write\_back, MAX) \\
        \textsc{round} += $1$
    }
    \Return \textsc{Dist}[$1:n$]
}
\caption{Breadth-First Search (BFS)}\label{alg:bfs_edgemap}
\end{algorithm}

\begin{algorithm}[b]
\fontsize{9pt}{12pt}\selectfont
\SetKwFor{IF}{If}{then}{endif}
\SetKwFor{While}{While}{do}{endwhile}
\SetKwFor{ParFor}{parallel for}{do}{endfch}
\SetKwFor{For}{for}{do}{endfch}
\SetKwFor{Function}{Function}{:}{}
\KwIn{$G=(V,E)$; \textsc{StartVertex}}
\KwOut{\textsc{BCValues}: distributed vertex value array}
Initialize distributed vertex value array: \textsc{NumPaths}[$1:n$], \textsc{BCValues}[$1:n$], \textsc{rounds}[$1:n$] \\
\textsc{round}$=1$ \\
\Function{f\_forward ( (u,v): Edge )}{
    \Return \textsc{NumPaths}[$u$]
}
\Function{write\_back\_forward ( v: Vertex, aggr\_value )}{
    \IF{\textsc{rounds}[$v$] == 0 or \textsc{rounds}[$v$] == \textsc{round}}{
        \textsc{NumPaths}[$v$].fetch\_add(aggr\_value) \\
        \IF{\textsc{rounds}[$v$] == 0}{
            \textsc{rounds}[$v$] = \textsc{round} \\
            \Return True
        }
    }
    \Return False
}
\Function{f\_backward ( (u,v): Edge )}{
    \Return \textsc{BCValues}[$u$]
}
\Function{write\_back\_backward ( v: Vertex, aggr\_value )}{
    \IF{\textsc{rounds}[$v$] == \textsc{round}}{
        \textsc{BCValues}[$v$].fetch\_add(aggr\_value)
    }
    \Return False
}
\Function{BC ( $G$, \textsc{StartVertex} )}{
    \DISTVERTEXSUBSET \textsc{frontier} = \{\textsc{StartVertex}\} \\
    \textsc{frontiers} = \{\textsc{frontier}\} \\
    \textsc{NumPaths}[\textsc{StartVertex}] = 1 \\
    \textsc{rounds}[\textsc{StartVertex}] = 1 \\
    \tcp{Forward Pass}
    \While{\textsc{frontier} is not empty}{
        \textsc{round} += $1$ \\
        \textsc{frontier} = \DISTEDGEMAP($G$, \textsc{frontier}, f\_forward, wb\_forward, add) \\
        \textsc{frontiers}.append(\textsc{frontier})
    }
    \tcp{Backward Pass}
    \ParFor{$i\in[1:n]$}{
        \textsc{BCValues}[$i$] = \textsc{NumPaths}[$i$] = $1/$\textsc{NumPaths}[$i$]
    }
    \While{\textsc{round} $>0$}{
        \textsc{frontier} = \textsc{frontiers}[\textsc{round}] \\
        \DISTEDGEMAP($G$, \textsc{frontier}, f\_backward, wb\_backward, add) \\
        \textsc{round} -= $1$ \\
    }
    \ParFor{$i\in[1:n]$}{
        \textsc{BCValues}[$i$] = \textsc{BCValues}[$i$] / \textsc{NumPaths}[$i$] - 1
    }
    \Return \textsc{BCValues}[$1:n$]
}
\caption{Betweenness Centrality (BC)}
\label{alg:bc}
\end{algorithm}

\clearpage

\begin{figure*}[tp]
    \centering
    \includegraphics[width=0.97\textwidth]{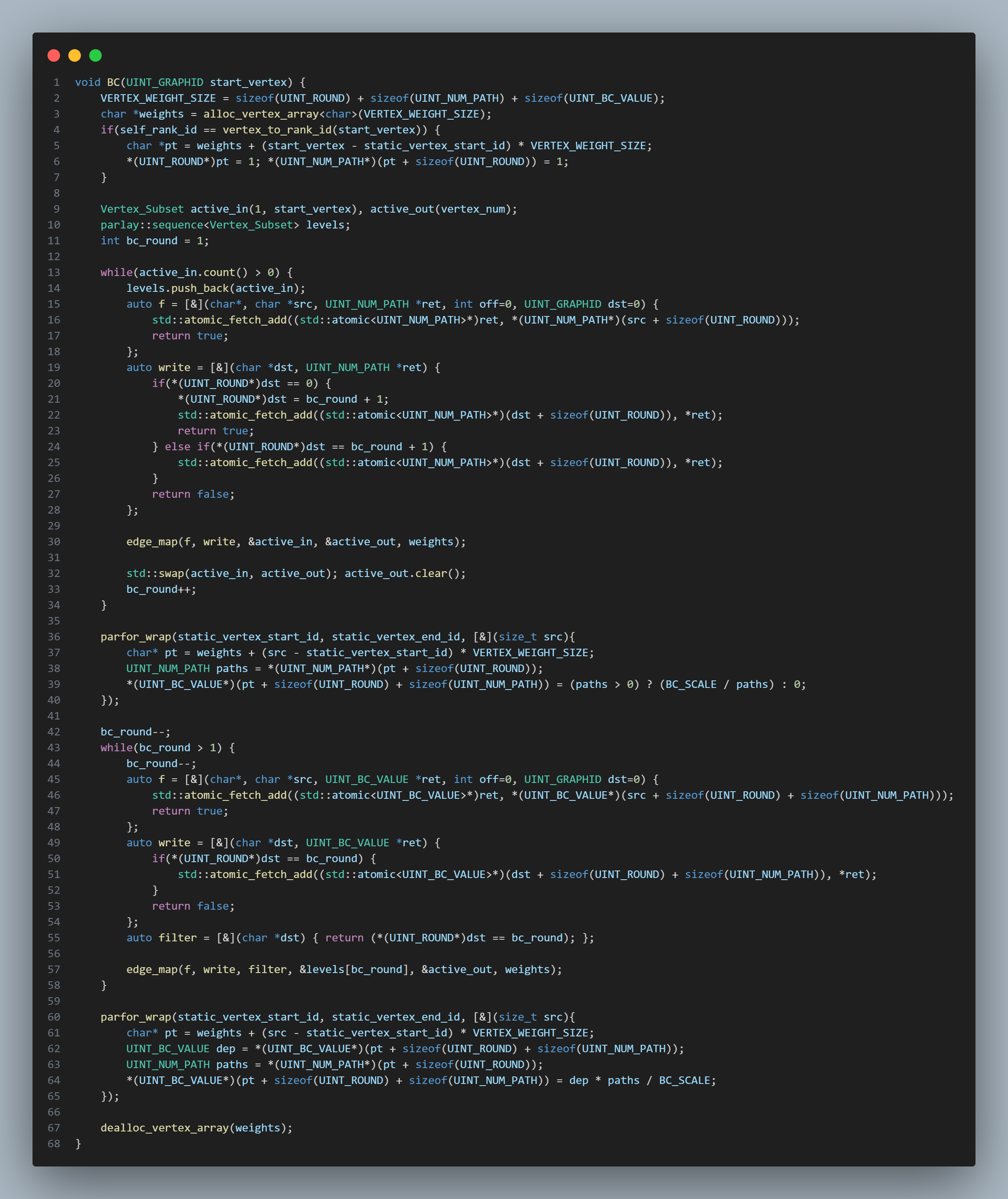}
    \caption{The actual C++ implementation of betweenness centrality (BC) using our \DISTEDGEMAP interface is highly concise---fewer than 70 lines of code.}
    \label{fig:bc_cplus_code}
\end{figure*}

\clearpage

\end{document}